\documentclass[a4paper]{article}

\usepackage{graphicx}%
\usepackage{multirow}%
\usepackage{amsmath,amssymb,amsfonts}%
\usepackage{amsthm}%
\usepackage{mathrsfs}%
\usepackage[title]{appendix}%
\usepackage{xcolor}%
\usepackage{textcomp}%
\usepackage{manyfoot}%
\usepackage{booktabs}%
\usepackage{algorithm}%
\usepackage{algorithmicx}%
\usepackage{algpseudocode}%
\usepackage{listings}%
\usepackage{amsthm}

\usepackage{authblk}

\theoremstyle{thmstyleone}%
\newtheorem{theorem}{Theorem}
\newtheorem{proposition}[theorem]{Proposition}%
\newtheorem{lem}[theorem]{Lemma}
\newtheorem{cor}[theorem]{Corollary}

\theoremstyle{thmstyletwo}%
\newtheorem{remark}{Remark}%

\theoremstyle{thmstylethree}%
\newtheorem{definition}{Definition}%

\title{Bisparse Blind Deconvolution through Hierarchical Sparse Recovery}

\author[1]{Axel Flinth}
\author[2,3]{ Ingo Roth}
\author[3]{Gerhard Wunder}

\affil[1]{Department of Mathematics and Mathematical Statistics, Umeå University, Sweden}
\affil[2]{Quantum Centre, Technology Innovation Institute, Abu Dhab, UAE}
\affil[3]{Dahlem Center for Complex Quantum Systems, Freie Universität Berlin, Germany}

\date{\today}

\usepackage{tikz}
\usepackage{xfrac} 

\definecolor{ingo}{rgb}{1,.2,.4}

\definecolor{axel}{rgb}{0,.6,.8}

\newcommand{\dd}{\mathrm{d}}

\usepackage{comment}

\usepackage{hyperref}
\usepackage{verbatim} 

\usepackage{graphicx}
\usepackage{caption}
\usepackage{subcaption}
\usepackage{amssymb}
\usepackage{ textcomp } 
\usepackage{color}
\usepackage{enumerate}
\usepackage{mathrsfs}
\usepackage{multicol}
\usepackage{hyperref}
\usepackage{amsopn}

\usepackage{amsmath}
\usepackage{amsthm}
\usepackage{amssymb}

\usepackage{multicol}

\newcommand{\C}{\mathbb{C}}

\newcommand{\F}{\mathbb{F}}
\newcommand{\K}{\mathbb{K}}
\newcommand{\N}{\mathbb{N}}

\newcommand{\R}{\mathbb{R}}

\newcommand{\calA}{\mathcal{A}}

\newcommand{\calC}{\mathcal{C}}

\newcommand{\calH}{\mathcal{H}}

\newcommand{\calK}{\mathcal{K}}
\newcommand{\calL}{\mathcal{L}}
\newcommand{\calM}{\mathcal{M}}
\newcommand{\calN}{\mathcal{N}}

\newcommand{\calT}{\mathcal{T}}

\newcommand{\abs}[1]{\left\vert #1 \right\vert}
\newcommand{\norm}[1]{\Vert #1 \Vert}
\newcommand{\snorm}[1]{\norm{#1}_{2 \to 2}}

\newcommand{\set}[1]{\left\lbrace #1\right\rbrace}
\newcommand{\sse}{\subseteq}

\newcommand{\polylog}{\mathrm{polylog}}

\newcommand{\prb}[1]{\mathbb{P}\left( #1 \right)}
\newcommand{\erw}[1]{\mathbb{E}\left( #1 \right)}

\usepackage{dsfont}

\newcommand{\geqsim}{\gtrsim}
\newcommand{\leqsim}{\lesssim}

\DeclareMathOperator{\supp}{supp}
\DeclareMathOperator{\id}{id}

\newcommand{\argmin}{\mathop{\mathrm{argmin}}}

\newtheorem{assumption}{Assumption}

\newtheorem*{res}{Main Result}

\theoremstyle{definition}

  \usepackage[margin=2.5cm]{geometry}

\begin{document}


\maketitle

\begin{abstract}
 The \emph{hierarchical sparsity framework}, and in particular the  HiHTP algorithm, has been successfully applied to many relevant communication engineering problems recently, particularly when the signal space is hierarchically structured.  In this paper, the applicability of the HiHTP algorithm for solving the bi-sparse blind deconvolution problem is studied. The bi-sparse blind deconvolution setting here consists of recovering $h$ and $b$ from the knowledge of $h*(Qb)$, where $Q$ is some linear operator, and both $b$ and $h$ are both assumed to be sparse. The approach rests upon lifting the problem to a linear one, and then applying HiHTP, through the \emph{hierarchical sparsity framework}. 
    Then, for a Gaussian draw of the random matrix $Q$, it is theoretically shown that an $s$-sparse $h \in \K^\mu$ and $\sigma$-sparse $b \in \K^n$ with high probability can be recovered when $\mu \geqsim s\log(s)^2\log(\mu)\log(\mu n) + s\sigma \log(n)$.

\end{abstract}

\section{Introduction}
Consider an underdetermined system of linear equations, i.e.  $y = Hx$ for an $H \in \K^{m,n}$ with $m \ll n$ ($\K$ denotes either $\R$ or $\C$). Such a system does of course not have a unique solution. However, if one assumes some structure of the ground truth $x_0$, it can still be recovered. This is the \emph{compressed sensing} paradigm\cite{FouRau2013}. The original example of such a structure is \emph{sparsity}, i.e. the assumption that only a few entries $x_0(i)$ of the ground truth are non-zero. In many applications such as communication engineering~\cite{Wunder2019_TWC}, there is even more prior knowledge on the sparsity, i.e. the signals 
are \emph{hierarchically sparse}.
\begin{definition}
\cite{hiHTP}
 Let $\mu, n, s, \sigma \in \N$. A tensor $$ w = \sum_{k\in [\mu]} e_k \otimes w_k \in \K^{\mu} \otimes \K^n  $$ is \emph{$(s,\sigma)$-sparse} if 
    \begin{itemize}
        \item At most $s$ of the blocks $w_i \in \K^n$ are non-zero.
        \item Each block $w_i$ is $\sigma$-sparse.
    \end{itemize}
    Here  $\{e_k\}_{k \in [\mu]}$ denotes the standard basis for $\K^\mu$ with entries $e_k(j) = 1$ for $k = j$ and $e_k(j) = 0$ otherwise. We will also simply refer to $w$ as \emph{hierarchically sparse}.
\end{definition}
In slightly different terms, we can think of a hierarchically sparse signal as a signal subdivided into $\mu$ fixed groups, each of sized $n$. Each such group is either zero, or $\sigma$-sparse, and in addition, at most $s$ groups are non-zero. See Figure \ref{fig:sparse_hisparse} for an illustration.

\subsection{HiHTP}
To recover a structured signal from linear measurements, so called \emph{hard thresholding pursuit} (HTP) algorithms can be applied. In short, the HTP entails alternately taking gradient descent steps to solve a least-squares problem and projections onto the set of structured signals. Originally developed for sparse vectors  \cite{foucart2011hard}, it was generalized to a considerably more inclusive model, under the moniker of \emph{model-based compressed sensing} in \cite{baraniuk2010model}. Applying the concept to hierarchically sparse signal yields the hierarchical hard thresholding pursuit, the \emph{HiHTP} (Algorithm \ref{alg:HiHTP})\cite{hiHTP}.

The main feature of HiHTP that distinguishes it from the generic model-based compressed sensing algorithm is that the projection step (line \ref{alg:HiHTP:TH} in Algorithm \ref{alg:HiHTP}) can be performed efficiently: To calculate the best $(s,\sigma)$-approximation of a tensor $w$, one first calculates the best $\sigma$-sparse approximation to each block $w_k$, which simply entails selecting the $\sigma$ largest entries of each block. One then selects the $s$ of the $\sigma$-sparse proxies 
with the highest $\ell_2$-norm. The computational complexity of such a projection is $O(\mu n)$, which is the same complexity as the projection onto the set of sparse signals. In addition, the first phase of approximations can be highly parallelised, which is not the case for the projection onto sparse signals.

\begin{algorithm}[tb]      
	\caption{(HiHTP)} 
	\label{alg:HiHTP}
	\begin{algorithmic} [1]
 		\Require Measurement operator $\calA: \K^\mu \otimes \K^n$, measurement vector $y$, block column sparsity $(s,\sigma)$
 		\State $w_0 = 0$ 
 		\Repeat
 			\State Calculate the support  $\Omega^{k+1}$ of the best approximation of $(w_k + \calA^\ast (y - \calA w_k))$ in the set of $(s,\sigma)$-sparse vectors. \label{alg:HiHTP:TH}
 			\State $w^{k+1} = \argmin_{z \in \K^{\mu} \otimes \K^{n}} \{ \|y - \calA z\|,\  \supp(z) \subset \Omega^{k+1} \}.$ \label{alg:HiHTP:LS}
 		\Until{stopping criterion is met at $\tilde{k} = k$}
   
 	\State \Return $(s,\sigma)$-sparse vector $w^{\tilde{k}}$
	\end{algorithmic}
\end{algorithm}

\begin{figure}
    \centering
    \includegraphics[width=0.5\linewidth]{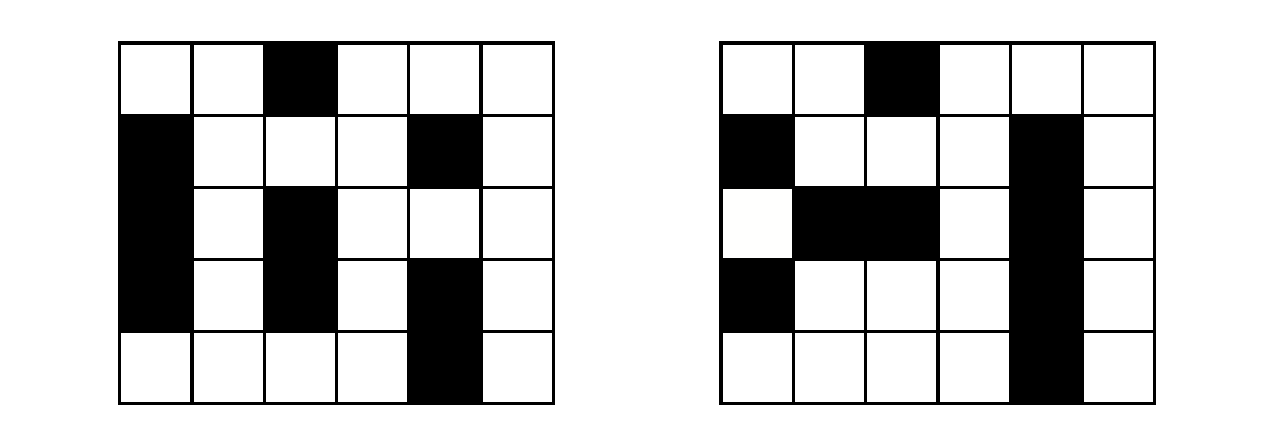}
    \caption{Two tensors which are both $9$-sparse. However, only the left one is hierarchically sparse -- it is $(3,3)$-sparse.}
    \label{fig:sparse_hisparse}
\end{figure}

As for any model-based compressed sensing algorithm \cite{baraniuk2010model}  the HiHTP will converge to a structured signal under a \emph{restricted isometry} condition. The standard restricted isometry constants are defined for sparse signals,
\begin{align*}
    \delta_s(A) = \min_{\substack{w\text{ }  s\text{-sparse} \\ 
         \norm{w}\leq 1}} \abs{\norm{\calA w}^2-\norm{w}^2}
\end{align*}
and the hi-sparse version similarly:

\begin{definition} \cite{hiHTP}
    Let $\calA: \K^\mu \otimes \K^n \to \K^m$. The $(s,\sigma)$-HiRIP constant of $\calA$ is given by
     \begin{align*}
         \delta_{(s,\sigma)}(\calA) =\min_{\substack{w\text{ }  (s,\sigma)\text{-sparse} \\ 
         \norm{w}\leq 1}} \abs{\norm{\calA w}^2-\norm{w}^2}.
     \end{align*}

\end{definition}

\begin{theorem} \label{th:HiRIP} \cite[Theorem 1]{hiHTP}
    If $\delta_{3s,2\sigma}<\sfrac{1}{\sqrt{3}}$, the sequence $w^k$ defined by HiHTP for a $y = \calA w_0 + e$ for a $(s,\sigma)$-sparse ground truth $w_0$ and error vector $e\in \K^m$ satisfies
    \begin{align*}
        \norm{w^k-w_0} \leq \rho^k \norm{w^0-w_0} + \tau \norm{e},
    \end{align*}
    with $\rho = (\sfrac{2\delta_{3s,2\sigma}}{(1-\delta_{2s,2\sigma}^2)})^{\sfrac{1}{2}}<1$, $\tau = \sfrac{5.15}{(1-\rho)}$.
In particular, when the  measurements $y$ are noise-free, $w^k$ converges towards $w_0$ at a linear rate.
\end{theorem}

     We will, deliberately imprecisely, say that an operator $\calA$ \emph{has the HiRIP}, if the HiRIP-constant is small enough for big enough sizes of $s$ and $\sigma$ to ensure the stable and robust recovery of any $(s,\sigma)$-sparse vector, as in Theorem \ref{th:HiRIP}.


In a series of papers \cite{hiHTP,roth2018hierarchical,flinth2021hierarchical,eisert2021hierarchical,Wunder2019_TWC,Wunder2023_CommNet} by, among others, the authors of this article  have shown that a number of important classes of operators have the HiRIP. 
\begin{itemize}
\item \underline{Gaussian operators.} Operators $A:\K^\mu \otimes \K^n\to \K^m$ whose matrix representations have i.i.d Gaussian entries, have the HiRIP with high probability when $m \geqsim s \log(\mu) +s\sigma \log(n)$, where $\geqsim$ refers to an inequality up to a universal multiplicative constant.
\item \underline{Kronecker products.} If $A\in \K^{m_0,\mu}$ and $B \in \K^{m_1,n}$ have  $s$- and $\sigma$-RIP constants $\delta_s(A), \delta_\sigma(B)$, respectively, the Kronecker product $A\otimes B: \K^\mu \otimes \K^n\to \K^{m_0} \otimes \K^{m_1}$ has $(s,\sigma)$-HiRIP constant smaller than $\delta_s+\delta_\sigma + \delta_s\delta_\sigma$.
\item \underline{Hierarchical operators.} A \emph{hierarchical measurement operator} is an operator $\mathcal{H}:\K^\mu \otimes \K^n\to \K^{m_0} \otimes \K^{m_1} $ defined as a mixture of matrices $(B_\ell)_{\ell \in [\mu]}$, $B_\ell \in \K^{m_1,n}$ as follows:
    \begin{align*}
        \calH(w)_k = \sum_{\ell \in [\mu]}a_{k\ell} B_\ell w_\ell
    \end{align*}
    for some matrix $A\in \K^{m_0,n}$. If $A$ has the $s$-RIP constant $\delta_s(A)$, and all $B_\ell$ have $\sigma$-RIP constants smaller than $\delta$, the $(s,\sigma)$-HiRIP constant of $\calH$ is not larger than $\delta_s(A) + \delta + \delta_s(A) \cdot \delta$. 
\end{itemize}

The purpose of this article is to prove that a new class of operators has the $(s,\sigma)$-RIP: namely, \emph{blind deconvolution operators}.

\subsection{Blind deconvolution}
The \emph{blind deconvolution problem} is the problem of determining a filter $h \in \K^\mu$ and a signal $x\in \K^\mu$ from its (circular) convolution, i.e.
\begin{align*}
    \K^\mu \ni y = h*x, \quad y(\ell) = \sum_{k\in [\mu]} h(\ell-k)x(k).
\end{align*}
It is standard to lift this bilinear recovery problem to a linear one (a detailed derivation will be presented later). The resulting problem is severely underdetermined -- the $\mu$ entries of $y$ will not be enough to uniquely identify the $\mu^2$-dimensional $h \otimes x$. Even the fact that $h\otimes x$ is a one-rank tensor is not enough -- additional structural assumptions on $h$ and $x$ are necessary. In this paper, we will assume that $h$ is $s$-sparse, and that $x$ can be written as $x=Qb$ for a known matrix $Q\in \K^{\mu,n}$, and $b\in \K^n$ a $\sigma$-sparse vector. Hence, we are confronted with the \emph{bi-sparse} version of the blind deconvolution problem.

These sets of assumptions can be motivated using the following simple model of mobile communications:  Alice and Bob are communicating with each other. Bob would like to send a message $b \in \K^n$ to Alice. To do so, he first linearly encodes his message in a signal $Qb \in \K^\mu$, and then sends it over a wireless channel to Alice. Due to scattering in the environment, Alice will not receive $Qb$, but rather a superposition of several time- and phase-shifted versions of the signal $Qb$ -- i.e., she will observe the convolution  $y =  h*(Qb)$ of $Qb$ with a filter $h \in \K^\mu$ {modelling the impulse response of the environment}. 
Since the scattered signal typically reaches the receiver antenna along a few, distinct paths, it is reasonable to assume that the filter $h$ is  \emph{sparse} \cite{BajwaEtAl:2010:CompressedChannelSensing}. The message $b$ can on the other hand be rendered sparse by system design. 

Using the lifting approach, the blind deconvolution problem then effectively turns into recovery problem of the tensor $h\otimes b$. Importantly, if both $h$ and $b$ are sparse, $h\otimes b$ is hierarchically sparse:  The blocks of $h\otimes b$ are given by  $h(k) b$, $k \in [\mu]$ -- they are all $\sigma$-sparse, and only the $s$ values for $k$ for which $h(k)\neq 0 $ corresponds to a non-zero block. This makes it natural to apply the HiHTP algorithm to the problem. This was done by a subset of the authors in \cite{wunder2018secure}, even adding more layers of hierarchy (see Section \ref{sec:demixing}) due to multiple antennas and users with good empirical success . It was also nicely applied in \cite{Wunder2023_CommNet,wunder2024perfect} as a component of a secure communication scheme. In neither of the papers, theoretical guarantees were provided, however. \emph{To provide such is essential in the security context}. {If they are not given, one can e.g. not establish guarantees for provable correctness of the scheme in \cite{Wunder2023_CommNet,wunder2024perfect}.} 

\subsection{The main result} 
Let us  present a, somewhat simplified, version of the main result of this article.
\begin{res} Consider the operator $\mathcal{C}$ corresponding to the blind deconvolution measurement $(h,b) \mapsto h*(Qb)$, where $Q\in \K^{\mu,n}$ is a Gaussian matrix. If
   \begin{align*}
       \mu \geqsim s\log(s)^2\log(\mu)\log(\mu n) +s\sigma \log(n),
   \end{align*}
   $\calC$ with high probability have the $(s,\sigma)$-HiRIP.
 \end{res}

Disregarding $\log$-terms, the blind-deconvolution operators hence have the HiRIP under the same sample complexity assumptions as the considerably simple Gaussian ones. Counting degrees of freedoms in describing a $(s,\sigma)$-sparse vector, it is also the sample optimal one.

\begin{remark} It is important to consider that with regards to the bisparse recovery problem, a sampling complexity of $s\sigma$ measurement far from  optimal. Counting degrees of freedom in the two vectors $h$ and $b$, one would instead hope for $s+\sigma$ measurements to be enough. In other words, $h\otimes b$ has more structure than only being $(s,\sigma)$-sparse, that one intuitively would like to take advantage of. However, if one  does not make additional assumptions on $h$ and $b$, a sample complexity of $s\sigma\cdot \mathrm{polylog}(\mu,n)$ is actually the state of the art for feasible algorithms\cite{geppert2019sparse}. We will give a detailed explanation of this in Section \ref{sec:litrev}.
\end{remark}

\subsection{Extensions}
  {A feature of the hierarchical approach to sparse deconvolution problems is that it readily generalizes to more complex settings, e.g.\ arising in multi-user communications. } 
   In this paper also discuss how to solve a combined deconvolution and demixing problem using the hierarchically sparse framework: 
   Given $M$ measurements of the form
   \begin{align*}
       y_q = \sum_{p \in [N]} d_{q,p}  h^p*(Q_pb^p),
   \end{align*}
   where $d_{q,p}\in \K,  \ q\in [M], p \in [N]$ are `mixing factors', recover the collection of $N$ filter-signal pairs $(h^p,b^p), p \in [N]$. 
   This problem, e.g., appears when users are separated both through their spatial angles and their multipath delay as in \cite{Wunder2019_TWC}. Under the assumption that only $S$ filters messages are non-zero, we can interpret this again as a hierarchical sparse recovery problem (albeit now in three levels of hierarchy). 
   Using recent results from  \cite{flinth2021hierarchical}, we show that under the same assumptions on the individual $Q_p$ as above, HiHTP can recover all signals and all filters from  $M \geqsim S\log(N)$ mixtures.

\subsection{Outline of the paper}
We begin our paper with a detailed literature review of different computational approaches to the blind deconvolution problem in Section~\ref{sec:litrev}. In Section~\ref{sec:HiHTP}, we give a formal introduction to our approach and our assumptions, and also formally state the main result, Theorem~\ref{th:main_result}. {We also briefly discuss the multi-user setup there.} Section~\ref{sec:proof_main_result} is in its entirety devoted to proving the main result. In Section~\ref{sec:numerics}, we perform a small numerical experiment to validate our theory.

\paragraph{Notation} Most of the notation is either standard or will be introduced at first use. 
For a vector $v\in \F^n$, we write $v(i)$ for its $i$:th entry. 
We use $\norm{\, \cdot \,}$ to denote the Euclidean norm for both vectors and tensors, $\norm{\, \cdot \,}_F$ for the Frobenius norm of a matrix, and $\norm{\, \cdot \,}_{2\to2}$ for the induced $\ell_2\to\ell_2$-operator norm of a matrix.

\section{Literature Review} \label{sec:litrev}

In this section, we give a literature review of previous approaches to the blind deconvolution problem in general, and  sparse blind deconvolution in particular. The approaches can broadly be divided into two categories:

\begin{itemize}
    \item As was described in the intro, one can identify the bilinear map {$\gamma(h,b)=h*Qb$}
    on  $\K^\mu \times \K^n$ with a  linear map $\calC$ on $\K^\mu \otimes \K^n$. One can then minimize
    \begin{align*}
        K: \K^\mu \otimes \K^n \to \R, \quad w \mapsto K(w) ={\calK}(y, \calC(w),w).
    \end{align*} 
    where $\calK$ refers to some suitable loss function. This procedure is generally referred to as \emph{lifting}, and we will therefore refer to this class of methods as \emph{lifted} methods. HiHTP is within this class.
    
   \item One can also directly minimize a function of the form
    \begin{align*}
        L: \K^\mu \times \K^n \to \R, \quad (h,b) \mapsto L(h,b) =\calL(y,h*(Qb),h,b),
    \end{align*}
    defined in terms of a suitable loss function $\calL$. Let us call such methods \emph{direct}.

\end{itemize}

We will give a brief overview of the relevant results from the literature about methods from both categories. 
Throughout the discussion, it is instructive to remember that the optimal sample complexity (i.e. the number of needed observations for injectivity) is given by $\mu + n$ in the dense setting and $s+\sigma$ in the sparse setting \cite{kech2017optimal,li2017identifiability}.

\paragraph{Lifted approaches} As for lifted approaches, it is crucial to understand that the lifted ground truth signal {$h_0\otimes b_0$ or equivalently its matrix form $h_0b_0^*$}, {where $b_0^*$ denotes the Hermitian transpose of $b_0$,} simultaneously enjoys two structures: it has rank one \emph{and} is sparse (in a structured manner). Each of these structures can individually be exploited using convex regularizers -- if we take the low-rank property as the crucial one, the nuclear norm is the canonical regularizer \cite{ahmed2013blind,ling2017blind,jung2018blind}, whereas if the sparsity is taken as the crucial property, the $\ell_1$-norm \cite{ling_2015} or its relative, the $\ell_{1,2}$-norm \cite{flinth2018sparse} should be used. While the convex approaches are stable, mathematically pleasing and globally solveable by off-the-shelf algorithms, they are computationally intensive and, thus, often impractical. Furthermore, recovery guarantees are only known for generic subspace models.  
    Under similar assumption of \emph{known supports} as above, the nuclear norm approach recovers the ground truth when $\mu \geqsim \max(s,\sigma)\log(\mu)^2$ \cite{ling2017blind,jung2018blind,Chen2021}. However, not assuming know supports, that guarantee degenerates to $\mu \geqsim \max(\mu,n)\log(\mu)^2$. The sparse models ($\ell_1$ and $\ell_{1,2}$) recover the ground truth with high probability when $\mu \geqsim s\sigma \log(sn)\log(\mu)^2$ \emph{under the assumption that the support of $h$ is known} \cite{ling_2015,flinth2018sparse}. If the support of $h$ is not known, the guarantee again degenerates to $\mu \geqsim \mu\sigma \log(n)\log(\mu)^2$. 
    
    There also is a natural non-convex-way to solve the lifted problem: a gradient descent projected onto the set of (bi)-sparse \emph{and} low-rank matrices. As is thoroughly discussed in \cite{foucart2020jointly}, there is however no practical algorithm available to compute this projection. 
    The bi-sparse unit-rank projection is known as the \emph{sparse PCA} problem, that has been shown to be worst-case NP-hard, as well as under approximation and on average    \cite{magdon-ismail_np-hardness_2017, chan_approximability_2016, brennan_optimal_2019}. 
    A canonical way to circumvent this obstacle is to alternate between projections onto the two sets. This approach is for instance investigated in \cite{eisenmann2021riemannian}. There, a \emph{local} convergence guarantee is presented under optimal sample complexity {using the considerably simpler measurement model with a Gaussian linear map without the structure of the convolution.} A global one is not.
    
    In this context, \cite{bahmani2016near} should also be mentioned -- there the authors obtain global convergence in just two alternations steps, however by assuming a nested measurement structure tailor-made for a jointly low-rank and sparse setting, which is often not given. 

    Finally, in \cite{ahmed2018leveraging}, a different but related problem was considered. In the paper, they consider the case of \emph{several} convolutions $y_i = h*(Q_i b_i), i \in [N]$ are independently observed. Importantly, the authors do \emph{not} to assume that the support of $h$ is known, only that it is sparse in a basis  $B$ with a certain incoherence. However, they instead assume that the supports of the $b_i$ are known, and that the $Q_i$ are \emph{independently} Gaussian distributed. Under those assumptions, they achieve recovery (through a nuclear norm minimization approach) of both $h$ and \emph{all} $b_i$ already when $\mu \geqsim (\sigma + s\log(s)^2)\log(\mu N)^4$ and $N\geqsim \log(N\mu)$. If we drop the assumption of known supports of the $b_i$, the guarantee again degenerates, to  $\mu \geqsim (\mu + s\log(s))\log(\mu s)^2$.

  \paragraph{{Direct} approaches} Among {direct approaches}, \emph{alternating minimization}, or \emph{sparse power factorization}, is probably the most prominent. It was introduced in \cite{NetrapalliAlternating} as a means for solving the so-called \emph{phase-retrieval problem} and later adapted to blind deconvolution in \cite{Bresler2015blind,lee2016blind}. 
    Alternating minimization  generally refers to taking turns in solving the following two problems,
    \begin{align*}
        \operatorname*{minimize}_{h \text{ (sparse) }}\ L(h,b') \quad\text{and}\quad   \operatorname*{minimize}_{b \text{ (sparse) } }\ L(h',b).
    \end{align*}
    Since the convolution is linear in each argument, each subproblem is effectively a classical compressed sensing problem, and can be solved using a number of techniques, e.g. iterative hard thresholding \cite{foucart2011hard} or CoSAMP \cite{needell2009cosamp}. 
    
     As for the question of success of alternating minimization, the authors of \cite{Bresler2015blind,lee2016blind} give a recovery guarantee in a special setting: First, they use a generic subspace model, where $h$ has a representation $h=Bg$ with some $s$-sparse vector $g$, rather than being sparse itself, {and $B$ is assumed to be a matrix with random Gaussian entries.} 
     The measurement matrix $Q$ is also assumed to be Gaussian. 
     An additional important assumption is that  $h$ and $x=Qb$ are \emph{spectrally flat},\ i.e.\ have absolutely approximately constant Fourier transforms. 
     The spectral flatness assumption is not only a regularity assumption on the signals.  
     It is actively exploited in a projection step of the authors' algorithm. 
     The projection step is hard to perform exactly.  
     The authors resort to heuristics for this step. However, accepting this caveat, the authors prove convergence already when only observing $(s+\sigma)\log(\mu)^5$ of the entries in $h*x$, which is up to log-terms sample optimal. 
    
    In \cite{li2019rapid}, Li, Ling, Strohmer and Wei proposed a different  non-convex approach {for the related problem where $h$ and $b$ are in known subspaces}. Concretely, they assume that $h$ is the Fourier transform of a vector $g \in \C^\mu$ supported on its first $s$ entries (rather than being $s$-sparse), and $b=Bc$ for some $c\in \C^\sigma$ and $B$ a Gaussian. Needless to say, this variation of the problem amounts to a significant simplification. 
    For a non-convex, smooth, loss function including regularization terms, together with a careful initialization, they prove convergence.  
    They also assume a spectral flatness condition. This assumption is again built into the method, now {implicitly in that it governs the size of a constant in the loss function}. 
    They prove global convergence already when $\mu \geqsim \max(s,\sigma)\log(\mu)^2$. 
    Note that in our setting, where the supports of the vectors are unknown, this condition would read $\mu \geqsim \max(\mu,n)\log(\mu)^2$, which would be vacuous.

    In \cite{lee2017near}, a RIP-based recovery theory is formulated for a version of sparse power factorization. which does \emph{not} explicitly depend on spectral flatness. They show that if the measurement operator fulfils a isometry property \emph{restricted to the set of bisparse signals}, sparse power factorization will recover any pair of sparse vectors, if started near enough the ground truth. Note that the latter 'bisparse RIP' is strictly weaker than the $(s,\sigma)$-HiRIP, and they indeed show that Gaussian measurement fulfil it already when the number of measurements are on the order of $(s+\sigma)\polylog(\mu,n)$. Although not carried out in the paper, it is clear to the specialists that their proof idea, in tandem with results from \cite{KrahmerSuprema} (that we will apply ourselves), should yield a similar needed sample complexity for blind deconvolution.

    Their recovery guarantee is however only local. They do describe a initialization procedure that will work with an optimal amount of measurements under a spectral flatness condition. This spectral flatness was weakened in \cite{geppert2019sparse}, but not broken: If one does not make any additional assumption on $h$ and $b$, one will not get a good enough initialization with less than $s\sigma\polylog(\mu,n)$ measurements.

   \paragraph{Conclusion} {The findings of this literature review can be summarized as follows.}
   \begin{itemize}
      
       \item \emph{Lifted methods}. There are convex and non-convex approaches to solve the lifted problems. The convex ones converge globally and are simple to implement, but are computationally costly and cannot recover the signals at optimal sample complexity. Their convergence can only be guaranteed when either the filter or the signal a-priori has known support.
       
       The non-convex ones are fast and converge quickly, but only work sample-optimally when computationally intractable projection routines are assumed.  When heuristics and approximations of the projection steps are applied instead, only local convergence results are known for practical measurement models.

        \item \emph{Direct methods}. These operate at optimal sample complexity locally. However, if initialization is taken into account,  additional assumptions on $h$ and $b$ are needed for an optimal sample complexity. If these are not made, $s\sigma\polylog(\mu,n)$ measurements are again needed.
   \end{itemize}

    We in particular conclude that to achieve an optimal sample complexity ($(s+\sigma) \cdot \mathrm{polylog}(\mu,n)$), one either needs to settle for local convergence guarantees, make additional assumptions on the signals and filters, or use infeasible projection steps. Global convergence without additional assumptions on $h$ and $b$ using a feasible projection steps, $s \sigma \mathrm{polylog}(\mu,n)$ is the best known complexity. In this sense, the results we present in this article HiHTP are as good as the state of the art.

One should note that HiHTP, as a lifted method, is not as memory and computationally efficient as certain direct methods (which only needs to update estimates of $h$ and $b$). However, one should also note that the best known initalization methods, discussed in e.g. \cite{lee2017near}, are essentially as heavy in both memory and computational complexity as one step of HiHTP. Since HiHTP typically converges quite fast, the method is hence not significantly more complex in an asymptotic sense.

\section{Bisparse blind deconvolution with HiHTP} \label{sec:HiHTP}

Let us first agree on some notation. As outlined above, we are interested in the blind deconvolution program, i.e. recovering $h \in \K^\mu$ and $x\in \K^\mu$ from the convolution $h*x$. We thereby understand the convolution as circular, i.e. 
\begin{equation}
    (h \ast x)(\ell) = \sum_{k \in [\mu]} h(\ell - k)x(k)\, ,
\end{equation}
 Here, $[\mu]$ is the set of remainder classes modulo $\mu$, i.e. the set of non-negative integers from $0$ to $(\mu-1)$ equipped with the addition modulo $\mu$. 
 In the following also all indices are understood using this cyclic identification.

The map $(h,b) \mapsto h*(Qb)$ is bilinear, and can therefore be lifted to a linear map $ \calC : \K^n\otimes \K^\mu \cong (\K^n)^\mu \to \K^\mu$ by linearly extending
\begin{equation}\label{eq:calCDef}
    h \otimes b \mapsto \calC(h\otimes b) = h*(Qb)\, .
\end{equation}
This means that we can interpret the blind deconvolution problem as a linear reconstruction problem for the \emph{lifted vector} $h \otimes b \in \R^\mu \otimes\R^n$.  As was argued in the introduction, if $h$ is $s$-sparse and $b$ is $\sigma$-sparse, $h\otimes b$ will be $(s,\sigma)$-sparse. It will hence be recoverable from $y =\calC(h \otimes b)$ provided $\calC$ has the HiRIP.

To prove the HiRIP, we will use the following model for the linear map $Q$.

\begin{assumption} \label{ass:measurement}
    The matrix $Q \in \K^{\mu,n}$ can be decomposed as follows 
    \begin{align*}
        Q=UA
    \end{align*}
    where
    \begin{enumerate}
        \item $A \in \K^{m,n}$ is a matrix with $\sigma$-RIP constant $\delta_\sigma(A)<1$.
        \item   $U\in \K^{\mu,m}$ is a matrix with columns $(\gamma_j)_{j \in [m]}$, whose entries $\gamma^{-\sfrac{1}{2}}\gamma_{j}(i)$ are
        \begin{itemize}
            \item centered, i.e. $\erw{ \gamma_j(i)}=0$.
            \item independent (over $i$ and $j$),
            \item normalized, i.e. $\mathbb{E}(\abs{\gamma_j(i)}^2)=1$.
            \item sub-Gaussian variables, i.e. that there exists a number $R$ so that $\erw{\exp(\tfrac{\abs{\gamma}^2}{R^2})}\leq 2$.
        \end{itemize}
    \end{enumerate}
\end{assumption}

\begin{remark}
    \item A more common definition of subgaussianity of a random variable $X$ is a tail estimate
    \begin{align*}
        \prb{\abs{X} > t} \geq 2 \exp(-\tfrac{t^2}{\widetilde{R}^2}).
    \end{align*}
    In fact, this notion is equivalent to the we use, with $R\sim \widetilde{R}$ \cite[Prop. 2.5.2]{highdimprob}. 
    \end{remark}
    \begin{remark}
     The infimum of all $R$ for which $\erw{\exp(\tfrac{\abs{\gamma}^2}{R^2})}\leq 2$ is \emph{the subgaussian norm} of $\gamma$, $\norm{\gamma}_{\psi_2}$. This is really a norm, so that e.g. $\norm{\lambda\gamma}_{\psi_2} = \abs{\lambda}\norm{\gamma}_{\psi_2}$ for $\lambda \in \K$.
     \end{remark}
\begin{remark} Gaussian variables (in $\C$ and $\R$) are subgaussian, with $\norm{X}_{\psi_2} \sim \erw{\abs{X}^2}$. Another example of subgaussian variables are bounded variables - if $\abs{X}\leq \theta$ almost surely, $\norm{X}_{\psi_2}\leq \theta^2$.

\end{remark}

Let us stress that the above assumptions entails simply choosing $Q$ as a Gaussian matrix. Indeed, choosing $m=n$ and $A=\id$, $A$ trivially has $\delta_\sigma(A)=0$ for all $\sigma$, and $Q=U$ becomes Gaussian. The model is hence \emph{more general} than what is usually considered in the literature.

To understand the intuition of the more complicated structure $Q=UA$, let us  rewrite
\begin{align}
    \calC(h \otimes b ) = \sum_{k \in [\mu]} h_k e_k*(Qb) = \sum_{k \in [\mu]} e_k*(UA(h_k b)) \label{eq:blind_deconvolution}.
\end{align}
Applying $\calC$ amounts to first applying the operator {$Q$} to the vector $b$ and subsequently mixing the shifts of the resulting signal $e_k~*~(Qb), k \in [\mu]$ using the scalars $h_k$ as weights. 
Our {standard RIP assumption} on $A$ ensures that, if we had access to the measurement $A(h_kb) \in \R^m$, $k \in [\mu]$ (before applying $U$), we could recover $h_kb, k \in [\mu]$ (and thus both $h$ and $b$) using standard compressed sensing methods. 
We can however only access them in observing their mixture  \eqref{eq:blind_deconvolution}. 
Being able to recover the different summands of this mixture requires sufficent `incoherence' of the shifts. 
We achieve this with the operator $U$: it embeds the vectors $A(h_kb)$ into a higher-dimensional space {rendering the shifts incoherent}.

{From the point of view of applications in wireless communication, we can think of $Q$ as 
applying first a compressive encoding of 
a sparse message $b \in \R^n$ giving rise to the `code word' $Ab$. 
The code word is subsequently again encoded using $U$ for transmission via the channel $h$ into a sequence $UAb$.}

  {The nested model for $Q$ also has a nice theoretical consequence. As we noted before, there is no need for $Q$ to have this nested structure -- choosing $A=\id$ is completely viable. Using a non-trivial matrix $A$, however, typically allows one to considerably reduce the number of parameters needed to describe the measurement map.  
  The product $UA$ is specified using only $\mu\cdot m + m\cdot n$ parameters, whereas without compression, $Q \in \K^{\mu,n}$ has $\mu\cdot n$ degree of freedoms. 
  Since typically $n \gg \sigma$ and the $\sigma$-RIP of $A$ can be guaranteed already when $m \geqsim \sigma\log(n)$, the former number can be considerably smaller than the latter for small values of $\sigma$. 
  Thus, our low-rank decomposition of $Q$ constitutes a 
  significant de-randomization of the measurement ensembles without considerably complicating the proof. 
  With other more structured standard compressed sensing matrices for $A$ such as sub-sampled Fourier matrices, the complexity of specifying $Q$ and the computational complexity of the algorithm can be further reduced. } Finally, it will allow us to slightly improve the sample complexity bound: A term $\log(\mu n)$ can be turned into $\log(\mu m)$.

We are now in a position to formulate the main result.
\begin{theorem} \label{th:main_result}
Fix $\delta_0 \in (0,1)$ and let $\epsilon>0$. {Suppose $Q=UA$ fulfils the  Assumption~\ref{ass:measurement}} and it holds that 
\begin{align}
      \mu \geqsim  (s\log(s)^2\log(\mu m)\log(\mu) + s\sigma \log(n)) \cdot \delta_0^{-2} \cdot \max(1,\log(\epsilon^{-1}))\,.
\label{eq:mubound}
\end{align}
Then, {the map $\calC$ defined in \eqref{eq:calCDef}} has
\begin{align*}
    \delta_{(s,\sigma)}(\calC) \leq (1+\delta_\sigma(A))^2\delta_0 + \delta_\sigma(A) 
\end{align*}
with a probability at least $1-\epsilon$.
\end{theorem}

In the special case of $Q$ being Gaussian, {$A=\id$, $m=n$}, the bound on the number of measurements grows as $s\log(s)^2\log(\mu)\log(\mu n)+ s\sigma \log(n)$. As was discussed in Section \ref{sec:HiHTP}, this sample complexity is essentially as good as the best known in the setting that we consider.  

We postpone the technical proof to the next section.  {Let us instead} {conclude this section by extending our approach to a combined deconvolution and demixing problem. }

\subsection{Deconvolution and demixing.}
{An attractive feature of the hierarchically sparse approach to the sparse blind deconvolution problem is its generalizability to more complicated problems, where the measurements arise from (potentially multiple) convolutions and further mixing (weighted superpositioning) of the signal. 
To illustrate this point, we consider the combined \emph{blind sparse deconvolution and demixing problem}: Recover $N$ filter-signal pairs $(h^p,b^p) \in \K^\mu\times \K^n$ from $M$ weighted mixtures of the convolutions $h^p*(Qb^p)$, 
\begin{align}
    y_q = \sum_{p \in [N]} d_{q,p} h^p *(Q_pb^p), q \in [M]\,, 
    \label{eq:demixing}
\end{align}
with mixing weights $d_{q,p} \in \K$.}

\paragraph{Motivation}
This problem is directly motivated by extending the previous  communication model to a multi-user scenario: {We imagine $N$ Alices simultaneously transmitting messages $b^p$, $p\in [N]$ over channels with impulse responses $h^p$. They also modulate their transmission, meaning that in timeslot $q$, Alice $p$  sends a message $d_{q,p}b^p$ rather than $b^p$. If the Alices do so over $M$ slots, Bob will receive a signal of exactly the form \eqref{eq:demixing}.}

{Another interesting, but a slightly less direct, route to motivate the model is a MIMO (Multiple-Input-Multiple-Output) scenario. In this setting, Bob has $M$ antennas arranged in an array. The collective response of the antennas of a single incoming wavefront $v(\ell)\in \K$ from a direction {at angle} $\theta$ is  given by a vector $v(\ell) d^\theta \in \K^M$. Hence, if the scattered transmitted signals $x$ from a user arrives with delays $k$ from directions $\theta_k$, the collective response of the antennas is }
 \begin{align*}
     y_q(\ell)  = \sum_{k \in [\mu]} h(k) x(\ell-k) d^{\theta_k}_q\, .
 \end{align*}
If several users simultaneously send signals $x^p$ over {channels} with {impulse responses} $h^p$, the antenna will measure a superposition of such measurements, 
 \begin{align*}
     y_q(\ell) = \sum_{p \in [N]} \sum_{k \in [\mu]}  h(k) x(\ell-k) d^{\theta_k^p}_q\, .
 \end{align*}
 {Let us assume that the angle for each transmitting user to Bob stays constant for a transmission. 
 Then, the wavefronts of one user are arriving from the same angle, and we have } 
 \begin{align*}
      y_q(\ell) = \sum_{p \in [N]} \sum_{k \in [\mu]} h(k) x(\ell-k) d^{\theta^p}_q = \sum_{p \in [N]} d^{\theta^p}_q (h^p * x^p)(\ell)\, .
 \end{align*}
With $x^p =Q_pb^p$ and $d_{q,p}=d(\theta^p)_q$, we arrive again at the form of \eqref{eq:demixing}.

\paragraph{{Multi-level hierarchical sparsity}} 

The hierachical sparsity model is very simple to extend to more than two levels.
\begin{definition}
    Recursively, we define a $K$-tensor $w \in \K^{n_0} \otimes \K^{n_1} \otimes \dots \otimes \K^{n_{k-1}}$ to be $(s_0, \dots s_{k-1})$-sparse if it can be written as
    \begin{align*}
        \sum_{\ell \in [n_0]} e_\ell \otimes w_\ell
    \end{align*}
    for $(s_1, \dots, s_{k-1})$-sparse tensors $w_\ell \in\K^{n_1} \otimes \dots \otimes \K^{n_{k-1}}$, out of which at most $s_0$ are non-zero.
\end{definition}
The HiHTP algorithm naturally extends to the multilevel setting, with a similar HiRIP-based recovery guarantee. We refer to \cite{eisert2021hierarchical} for details.

The deconvolution and demixing problem can now easily be formulated as a hierarchically structured recovery problem. 
We view the collection $ (h^p \otimes b^p)_{p \in [N]} \in (\K^\mu \otimes \K^n)^N $ as a three-tensor
\begin{align*}
    X = \sum_{p \in [N]}e_p \otimes h^p \otimes b^p \in \K^N \otimes \K^\mu \otimes \K^n\, .
\end{align*}
If we assume that each message $b^p$ is $\sigma$-sparse, each filter is $s$-sparse, and that only $S$ filter-message tensors $h^p \otimes b^p$ are non-zero, the ground truth signal $X$ is $(S,s, \sigma)$-sparse.  {In the above communication model, such an assumption corresponds to assuming \emph{sporadic user activity} with only a subset of possible users actively sending at a given instance in time.  
Such sporadic traffic is well-motivated in today's and future machine-type communications.}

The blind deconvolution and demixing model can hence be treated as a hierarchical sparse recovery problem. As such, we may use the HiHTP to recover it. 
Defining $d^p = (d_{q,p})_{q \in [M]} \in \K^M$, and $\calC_p$ as the lifted convolutions associated to filter-message pair $p$, the collective measurement operator takes the form
$\calM : \K^N \otimes \K^\mu \otimes \K^n \to \K^M \otimes \K^\mu$, 
\begin{align}\label{eq:combinedCalMDef}    
    \sum_{p \in [N]}e_p \otimes w_p \mapsto \sum_{p \in [N]} d^p \otimes \calC_p(w_p)\, .
\end{align}
These type of operators {for arbitrary operators $\calC_p$ having HiRIP} were analyzed by the authors of the article in \cite{flinth2021hierarchical}. 
{The results \cite[Theorem~2.1]{flinth2021hierarchical} state that such a \emph{hierarchical} measurement operator inherits a HiRIP from the RIP of its constituent operators and it holds that}
\begin{align*}
    \delta_{(S,s,\sigma)}(\calM) \leq \delta_S(D) + \sup_p \delta_{(s,\sigma)}(\calC_p) +  \delta_S(D)  \cdot \sup_p \delta_{(s,\sigma)}(\calC_p),
\end{align*}
where $D=[d_p]_{p\in [N]} \in \K^{M,N}$.
Hence, if $\delta_S(D)$ and each $\delta_{s,\sigma}(\calC_p)$ is small, the entire measurement $\calM$ will have an $(S,s,\sigma)$-HiRIP. Let us formulate this and {its immediate consequences} as a corollary.
\begin{cor}[{Blind Deconvolution and demixing guarantee}]
(i) Assume that $D\in \K^{M,N}$ has $\delta_S(D)<1$ and $\calC_p$ fulfils $\sup_{p \in [N]} \delta_{(s,\sigma)}(\calC_p) =: \delta_\calC<1$. 
Then, the $\calM$ defined in \eqref{eq:combinedCalMDef} obeys
\begin{align*}
    \delta_{(S,s,\sigma)}(\calM)\leq \delta_S(D) + \delta_\calC +  \delta_S(D)  \cdot \delta_\calC .
\end{align*}

(ii) Assume that the $\calC_p$ are equal to one operator constructed according to Assumption \ref{ass:measurement}, and $D \in \K^{M,N}$ is Gaussian. Suppose
\begin{align*}
    M &\geqsim S\log(N) \\
    \mu &\geqsim s\log(s)^2\log(\mu m)\log(\mu) + s\sigma \log(n),
\end{align*}
then with high probability, HiHTP will recover each $(S,s,\sigma)$-sparse ground truth $X$ from the measurements $\calM(X)$.

\end{cor}
\begin{remark}
The assumption of identical $\calC_p$  is not essential -- it merely simplifies the formulation of the result. Using a union bound together with Theorem \ref{th:main_result}, {we can also address the case where the $\calC_p$ are independently chosen and establish a recovery guarantee for}
\begin{align*}
    \mu \geqsim (s\log(s)^2\log(\mu m)\log(\mu) + s\sigma \log(n))\log(N)\, .
\end{align*}
{The additional $\log(N)$ factor comes from applying a union bound over the $N$ operators $\calC_p$.}  Although this bound does not suggest it, letting the constituent operators $\calC_p$ be different, and in particular uncorrelated, may in the general hierarchical measurement case assist the recovery. We refer to \cite{flinth2021hierarchical} for a detailed discussion on this matter.
\end{remark}

\section{Proof of Theorem \ref{th:main_result}} \label{sec:proof_main_result}
{The proof of Theorem \ref{th:main_result} proceeds in five steps. First, we rewrite the problem, partly factoring out the operator $A$. Then, we follow a route taken many times in the literature before, for instance in \cite{KrahmerSuprema,rudelson2008sparse} and in particular presented in the textbook \cite{FouRau2013}: We identify the HiRIP constant as the supremum a random process as in \cite{KrahmerSuprema}. Using the results of said paper, we reduce the problem to the estimation of so-called \emph{$\gamma_\alpha$-functionals} of a set of matrices, which in turn boils down to the estimation of covering numbers. The third step consists in providing a simple estimate of the latter, and the fourth in applying \emph{Maurey's empirical method} to derive a {tighter} bound {for a certain regime}. {Both bounds} are finally assembled to complete the proof and obtain the sample complexity {in a fifth step}. Although the proof follows well-known beats, the particular geometry of the set of $(s,\sigma)$-sparse vectors and the structure of our measurement model requires careful and non-trivial modifications of the arguments.}

\subsection{ Preliminaries}
 We start out small and provide an explicit form of the lifted convolution operator. To do that, we need some further notation.
\begin{definition}
For $\ell \in [\mu]$, let $S_\ell:\K^\mu \to \K^\mu$ be the shift operator
        \begin{align*}
            [S_\ell x]_k = x_{k+\ell}\, .
        \end{align*}
 We define $\calH : \K^\mu  \otimes  \K^n \to \K^\mu$ through
        \begin{align*}
            w= \sum_{k \in [\mu]} e_k \otimes w_k \ \mapsto\ \calH(w) = \sum_{k \in [\mu]}S_k  Qw_k\, .
        \end{align*}
\end{definition}
The claim is now that \emph{up to reflections}, the lifted version of the convolution is given by $\calH$.
\begin{lem} \label{lem:lifting}
    Define the reflection operator $R: \R^\mu \to \R^\mu$  through
    \begin{align*}
        [Rx]_k = x_{-k},
    \end{align*}
    and extend it to $\R^\mu \otimes \R^n$ through linear extension of $R(h \otimes b) = R(h) \otimes b$. Then, for $h \in \R^\mu$ and $b \in \R^n$,
    \begin{align*}
        \calC = \calH \circ R
    \end{align*}

\end{lem}
\begin{proof}
  It suffices to show the claimed equality on elements $h \otimes b$. Hence, it suffices to show that $[h*(Qb)]_{\ell} = \calH(Rh \otimes b)_\ell$ for each $\ell \in [\mu]$. Note that $Rh \otimes b = \sum_{k \in [\mu]} (Rh)_k e_k \otimes b =  \sum_{k \in [\mu]} h_{-k} e_k \otimes b$. Hence, for each $\ell\in [\mu]$ we have 
   \begin{align*}
        \calH(Rh \otimes b)_\ell &= \sum_{k \in [\mu]} [S_k(h_{-k} Q b)]_\ell =  \sum_{k \in [\mu]} h_{-k}[Qb]_{k + \ell} 
        = \sum_{j \in [\mu]} h_{\ell -j}[Qb]_j  = [h*(Qb)]_{\ell},
    \end{align*}
    which was what to be shown.
\end{proof}

Since the reflection is an isometry (on $\K^\mu$ as well as $\K^\mu \otimes \K^n$) and leaves the hierarchical sparsity structure intact, we can concentrate on proving the HiRIP for the $\calH$-operator.  Let us introduce further notational simplification.
\begin{definition}
  For $w = \sum_{k \in [N]} e_k \otimes w_k \in {\K^{\mu}\otimes \K^n}$ and $A \in \K^{m,n}$, we let $Aw$ denote the result of a block-wise application of $A$ to $w$, i.e.
    \begin{align*}
        Aw := \sum_{k \in [n]} e_k \otimes (Aw_k).
    \end{align*}
    
\end{definition}
\begin{definition} For $s,\sigma$, we define $T_{s,\sigma}$ be the set of $(s,\sigma)$-sparse vectors in $\K^\mu \otimes \K^n$ with norm less than or equal to $1$, and $T_\sigma \sse \K^n$ similarly. We furthermore define $AT_{s,\sigma}$ as the image of $T_{s,\sigma}$ under $A$, i.e.
   \begin{align*}
       AT_{s,\sigma} = \set{ \Sigma_{k \in [\mu]} e_k \otimes (Aw_k) \, \vert \, w = \Sigma_{k \in [\mu]} e_k \otimes w_k \in T_{s,\sigma}}\,.
   \end{align*}
\end{definition}
We now prove a slightly less trivial lemma than the ones above. It shows that we may partially factor out the matrix $A$ from the analysis. 
\begin{lem} \label{lem:factorization}
{For $Q=UA$, let $\widehat{\calH}: \K^\mu \otimes \K^m \to \K^\mu$ denote the operator
\begin{align*}
    v = \sum_{k\in [\mu]}e_k \otimes v_k\  \mapsto\  \widehat{\calH}(v) = \sum_{k\in [\mu]}S_k U v_k\,. 
\end{align*}}
Let further
\begin{align} \label{eq:redRIP}
    \Delta_{s,\sigma}(\widehat{\calH}) := \sup_{w \in AT_{s,\sigma}} | \|\hat {\mathcal H} (w) \|^2 -\|w \|^2|,
\end{align}

We then have
\begin{align*}
     \delta_{s,\sigma}(\calH)\leq\Delta_{s,\sigma}(\widehat{\calH})+\delta_{\sigma}(A)+ \Delta_{s,\sigma}(\widehat{\calH})\delta_{\sigma}(A)\,.
\end{align*}
\end{lem}

\begin{proof}
   Let $w \in \K^\mu \otimes \K^n$ be a normalized $(s,\sigma)$-sparse vector. Then, $Aw$ clearly is an element of $AT_{s,\sigma}$. Moreover, $\calH(w)=\widehat{\calH}(Aw)$. Hence,
   \begin{align*} \abs{\norm{\calH(w)}^2-\norm{Aw}^2} = \abs{\norm{\widehat{\calH}(Aw)}^2-\norm{Aw}^2} \leq \Delta_{s,\sigma}(\widehat{\calH})\norm{Aw}^2.
   \end{align*}
   Now, since each $w_k$ is $\sigma$-sparse
   \begin{align*}
     \abs{\norm{Aw}^2-\norm{w}^2} \leq \sum_{k\in [\mu]}\abs{\norm{Aw_k}^2- \norm{w_k}^2} \leq \delta_{\sigma}(A) \sum_{k\in [\mu]} \norm{w_k}^2 =\delta_{\sigma}(A) \norm{w}^2.
   \end{align*}
   Combining these two inequalities yields
   \begin{align*}
     \abs{\norm{\calH(w)}^2-\norm{w}^2} &\leq \abs{\norm{\calH(w)}^2-\norm{Aw}^2} +   \abs{\norm{Aw}^2-\norm{w}^2} \\
     &\leq \Delta_{s,\sigma}(\widehat{\calH})\norm{Aw}^2 + \delta_\sigma(A) \norm{w}^2 \leq \left(\Delta_{s,\sigma}(\widehat{\calH})(1+ \delta_\sigma(A)) +  \delta_\sigma(A)\right) \norm{w}^2,
   \end{align*}
   which is the claim.
\end{proof}
\begin{remark} {A careful inspection of the proof shows that we can let $\sigma$ be a multiindex {referring to a hierarchically sparse structure itself}. 
Then, it suffices that the compression matrix $A$ has a $\sigma$-\emph{Hi}RIP instead (where $\sigma$ is understood as a tuple of sparsities) to arrive at the analogous result. }
It is hence only a notational effort to extend the above results to the case where the signal $b$ is not only sparse, but hierarchically sparse. In the same manner, what follows should also be generalizable to the multilevel setting. We have chosen to stay in the bilevel regime mainly to increase readability. \label{rem:hisparse}
\end{remark}

\subsection{Suprema of empirical chaos processes}

The main tool of our argument will be a result from \cite{KrahmerSuprema}. It uses the so-called Talagrand's $\gamma_\alpha$-functionals. 
Since the definition of the functional is rather technical and not particularly insightful, we refer the interested reader to \cite[Definition ~1.2.5]{genericChaining}. 
For our purposes, we only need to know that for a set $T$ and a metric $d$, 
\begin{align*}
    \gamma_\alpha(M,d) \leqsim \int_0^\infty \log(\calN(T,d,u))^{\sfrac{1}{\alpha}} \, \dd u,
\end{align*}
where $\calN(T,d,u)$ denotes the covering numbers of the set $T$ \cite{genericChaining,dirksen2015tail} \cite[App.~A]{RauhutCirculant}, and $\alpha$ is a positive number. The aforementioned result is as follows.
\begin{theorem}(\cite[Theorem 3.1, see also Remark 3.4]{KrahmerSuprema}) \label{th:krahmerbound} Let $\gamma$ be a random vector as in Assumption \ref{ass:measurement} and $\calT$ a symmetric set of matrices in $\C^{\mu,\mu\cdot n}$. Define
\begin{align*}
d_{2\to 2} = \sup_{T \in \calT} \norm{T}_{2\to 2}, \quad d_{F} = \sup_{T \in \calT} \norm{T}_{F}
\end{align*}
and set
    \begin{align*}
        E &= \gamma_2(\calT, \norm{\cdot}_{2\to2})^2 + d_F \cdot \gamma_2(\calT,\norm{\cdot }_{2\to2}), \quad   U = d_{2\to2}(\calT)^2 \text{ and } \\
        V&= d_{2\to 2}(\calT) \left( \gamma_2(\calT, \norm{\cdot }_{2\to2}) + d_F(\calT)\right) 
    \end{align*}
Then
\begin{align*}
    \prb{\sup_{T \in \calT} \abs{\norm{T\gamma}^2 - \erw{\norm{T\gamma}^2}} > CE + t} \leq 2\exp\left(- \min\left\{\tfrac{t^2}{V^2},\tfrac{t}{U}\right\}\right).
\end{align*}
\end{theorem}

We now, as in \cite{KrahmerSuprema}, rewrite our measurement $\calH(w)$ as a product of a matrix and a random vector.

\begin{lem} \label{ref:umschreibung}
Let 
\begin{align*}
    \calA : \C^\mu \otimes \C^n \to \C^{\mu,\mu \cdot n}, w \mapsto \left(\tfrac{1}{\mu^{\sfrac{1}{2}}} w_{j-k}(i)\right)_{k \in [\mu], (j,i)\in [\mu]\times [n]}.
\end{align*}
Then,
\begin{align*}
     \Delta_{s,\sigma}(\widehat{H}) = \sup_{T \in \calA(AT_{s,\sigma})}\abs{\norm{T\gamma}^2 - \erw{\norm{T\gamma}^2}},
\end{align*}
\end{lem}
\begin{proof}
We have for each $k \in [\mu]$
\begin{align*}
    (\widehat{\calH}w)(k) &= \sum_{\ell \in [\mu]} (S_\ell Uw_\ell)(k) =  \sum_{\ell\in [\mu], i \in [n]} (S_\ell u_i)(k) w_\ell(i) = \sum_{\ell\in [\mu], i \in [n]}  u_i(\ell+k) w_\ell(i)  \\
    &= \lceil j = \ell + k \rceil = \tfrac{1}{\sqrt{\mu}}\sum_{j \in [\mu], i \in [n]} w_{j-k}(i) \gamma_i(j) = (\calA(w)\gamma)(k),
\end{align*}
i.e. $\widehat{\calH}(w)=\calA(w)\gamma$. Now we calculate
\begin{align*}
    \erw{\norm{\calA(w)\gamma}^2} &= \erw{\mathrm{tr}(\gamma^*\calA(w)^*\calA(w) \gamma)} = \erw{\mathrm{tr}(\calA(w)^*\calA(w) \gamma\gamma^*)} \\
    &= \mathrm{tr}(\calA(w)^*\calA(w)) = \norm{\calA(w)}_F^2,
\end{align*}
where we in the penultimate step used that $\erw{\overline{\gamma_i}\gamma_j}=\delta_{ij}$. Now it is only left to note that
\begin{align*}
    \norm{\calA(w)}_F^2 = \sum_{k,j \in [\mu], i \in [n]} \tfrac{1}{\mu} \abs{w_{j-k}(i)}^2 = \norm{w}^2.
\end{align*}
\end{proof}

To bound the constant $\Delta_{s,\sigma}(\calH)$, we now need to estimate the parameters $E,V$ and $U$ in Theorem \ref{th:krahmerbound}.  To start, we have the following Lemma.
\begin{lem} \label{lem:distances}
    Define the norm
    \begin{align*}
        \norm{v}_X = \left(\sum_{i \in [m]} \sup_{\theta \in [0,2\pi]} \abs{\sum_{k \in [\mu]}v_k(i)e^{\iota k \theta}}^2\right)^{\sfrac{1}{2}}.
    \end{align*}
    on $\K^{\mu} \otimes \K^m$. Then, for $u,w \in \K^\mu \otimes \K^m$
    \begin{align*}
        \norm{\calA(u-w)}_{2\to 2} \leq \tfrac{2}{\sqrt{\mu}}\norm{u-w}_X.
    \end{align*}
\end{lem}
\begin{proof}
    Let us for convenience write $v=u-w$. Let us notice that $\calA(v)$ is a matrix consisting of $m$ Toeplitz blocks
    \begin{align*}
        \calA(v) = [N^0, \dots, N^m],
    \end{align*}
    with $N^i = \tfrac{1}{\mu^{\sfrac{1}{2}}}  (w_{j-k}(i) )_{j,k \in [\mu]}$. The operator norm of a Toeplitz matrix $\mathcal{T}~=~(\tau(j~-~k))_{j,k \in [\mu]}$ is well-known to be upper bounded by (see for example  \cite{Bottcher1999})
    \begin{align*}
        \snorm{\mathcal{T}}\leq 2\sup_{\theta \in [0,2\pi]}\abs{\sum_{k \in [\mu]} \tau(k)e^{\iota k\theta}}.
    \end{align*}
    We may therefore estimate
    \begin{align*}
        \snorm{N^i} \leq 2\sup_{\theta \in [0,2\pi]}\abs{\tfrac{1}{\sqrt{\mu}}\sum_{k \in [\mu]} v_{k}(i)e^{\iota k\theta}}
    \end{align*}
    which together with the bound
    \begin{align*}
        \norm{\calA(v)\gamma} \leq \sum_{i \in [m]}\norm{N^i}\norm{\gamma_i} \leq  \left(\sum_{i \in [m]}\norm{N^i}^2\right)^{\sfrac{1}{2}} \norm{\gamma} 
    \end{align*}
    gives the claim.
\end{proof}

The above lemma shows that $\calN(\calA(AT_{s,\sigma}), \norm{\, \cdot \,}_{2\to 2},t)\leq \calN(AT_{s,\sigma},\tfrac{2}{\sqrt{\mu}}\norm{\,\cdot\,}_X,t)$, and similar for the diameters. The rest of the proof will therefore consist in estimating the quantities involving the $\norm{\, \cdot \,}_X$-norm. 

\subsection{A simple covering number bound} 
The estimation of the $\gamma_2$-functional boils down to bounding the covering numbers $\calN(AT_{s,\sigma},\norm{\, \cdot \,}_X, t)$. The following lemma gives us a first, simple, estimate of the latter.
\begin{lem} \label{lem:smallt}
    Let $u \in AT_{s,\sigma}$ have a representation $\check{u} \in T_{s,\sigma}$, i.e., $u = A\check{u}$. Then
    \begin{align*}
        \norm{u}_X \leq \sqrt{s} (1+\delta_\sigma(A))^{\sfrac{1}{2}}\norm{\check{u}}.
    \end{align*}
    In particular,
    \begin{align*}
        \log(\calN(\calA(AT_{s,\sigma}), \norm{\, \cdot \,}_{2\to 2},t))\leq s\log(\mu) + s\sigma \log(n) + 2s\sigma \log(1+\tfrac{4\sqrt{s}(1+\delta_{\sigma}(A))^{\sfrac{1}{2}}}{t\sqrt{\mu}}) \\
    \end{align*}
\end{lem}

\begin{proof}
    For each $\theta \in [0,2\pi]$ and $i\in [m]$, we have
    \begin{align*}
        \abs{\sum_{k \in [\mu]}u_k(i) e^{\iota k \theta}}
        \leq \sum_{k\in [\mu]}|u_k(i)|
        \leq \sqrt{s} \|u\|, 
    \end{align*}
    where we have used reverse norm-ordering $\norm{v}_1  \leq \sqrt{s}\norm{v}_2$ for an $s$-sparse  vector for the $(u_k(i))_{k\in\mu}$ with fixed $i$.
    To prove the first claim, it is now only left to note that since each $\check{u}_k$ is $\sigma$-sparse and $A$ has the $\sigma$-RIP, we have
    \begin{align*}
        \norm{u}^2 = \sum_{k \in [\mu]} \norm{A\check{u}_k}^2 \leq \sum_{k \in [\mu]} \left(1+\delta_\sigma(A)\right)\norm{\check{u}_k}^2 = (1+\delta_\sigma(A)) \norm{\check{u}}^2.
    \end{align*}
    The above estimate allows us to bound the covering number as follows. Let us set $\kappa = 2\sqrt{\tfrac{s}{\mu}}(1+\delta_\sigma(A))^{\sfrac{1}{2}}$. For each $(s,\sigma)$-sparse support $\Omega$, let $U_\Omega$ denote intersection of the unit ball with the signals supported on $\Omega$. Since the latter spaces in both the real and complex case have a dimension smaller than $2s\sigma$, it is well known that (see for example \cite[Appendix 2]{FouRau2013}) we have for each $t>0$
    \begin{align*}
        \calN(U_\Omega,\kappa \cdot \norm{\cdot}, t) \leq \calN(U_\Omega,\norm{\cdot}, \tfrac{t}{\kappa})\leq \left(1+\tfrac{2\kappa}{t}\right)^{2s\sigma}.
    \end{align*}
    The above means that there exists a $t$-net of $U_\Omega$ w.r.t $\kappa \norm{\cdot}$ of cardinality less than $\left(1+\tfrac{2\kappa}{t}\right)^{2s\sigma}$. By the inequality we just proved, and Lemma \ref{lem:distances}, such a net is however also a $t$-net of $\calA(U_\Omega)$ with respect to $\norm{\cdot}_{2\to2}$. Hence, since there are $\binom{\mu}{s}\binom{n}{\sigma}^s$ $(s,\sigma)$-sparse supports, we conclude
    \begin{align*}
        \calN(\calA(AT_{s,\sigma}),  \norm{\, \cdot \,}_{2\to 2}, t) \leq \binom{\mu}{s}\binom{n}{\sigma}^s\left(1+\tfrac{2\kappa}{t}\right)^{2s\sigma},
    \end{align*}
    which yields the second claim.
\end{proof}

The above estimate of the covering number is a little bit too crude for large values of $t$. To provide a better one, we will have to use more involved tools.

\subsection{The empirical method of Maurey}
Now we will use Maurey's empirical method to provide a stronger bound of the covering numbers $\calN(AT_{s,\sigma},\norm{\, \cdot \,}_X, t)$. The idea of this method is to, given $v \in AT_{s,\sigma}$ define random variables $\underline{v}$ that only take on finitely many, say $N$, values, and subsequently show that $\erw{\underline{v}-v}\leq t$. This shows in particular that one of the $N$ values $\underline{v}$ can attain is closer to $v$ than $t$ and, thus, $\calN(AT_{s,\sigma},\norm{\,\cdot \,}_X, t)\leq N$.

The first step towards defining the random variable $\underline{v}$ is to notice that $T_{s,\sigma}$ is contained in a scaled version of the $\ell_{1,2}$-unit ball $B_{1,2}$. Recall that  the $\ell_{1,2}$-norm of an element $\check{v} =\sum_{k \in [\mu]} e_k \otimes \check{v}_k \in \K^{\mu} \otimes \K^n$ is given by $\norm{v}_{1,2} =\sum_{k \in [\mu]} \norm{v_k}$. Now, if $\check{v} \in T_{s,\sigma}$, the vector $\nu = (\norm{\check{v}_k})_{k\in [\mu]}$ is $s$-sparse. Hence, the reversed norm ordering $\norm{w}_1\leq \sqrt{s}\norm{w}_2$ for $s$-sparse $w$ yields
\begin{align*}
    \norm{\check{v}}_{1,2} = \sum_{k \in [\mu]}\norm{\check{v}_k} =\norm{\nu}_1 \leq \sqrt{s}\norm{\nu}_2 = \sqrt{s}\norm{\check{v}}.
\end{align*}
Hence, $\check{v} \in \sqrt{s}B_{1,2}$, where $B_{1,2}$ is the unit ball with respect to the $\norm{\, \cdot \,}_{1,2}$-norm. Clearly, $\check{v}$ is also in the set
\begin{align*}
    \Sigma_\sigma^\mu = \set{ \sum_{k \in \mu} e_k \otimes \check{v}_k \, \vert \, \check{v}_k \text{ $\sigma$-sparse.}}
\end{align*}
so that every $v\in A(T_{s,\sigma})$ must lie in $A(\sqrt{s} B_{1,2} \cap \Sigma_\sigma^\mu)$. This insight leads to the following statement.
\begin{lem} \label{lem:l12ball}
    \begin{align*}
        \calN(\tfrac{1}{\sqrt{s}}A(T_{s,\sigma}),\norm{\,\cdot\,}_X,t) \leq \calN(A(B_{1,2} \cap \Sigma_\sigma^\mu), \norm{\,\cdot \,}_X, \tfrac{t}{2}). 
    \end{align*}
\end{lem}
\begin{proof}
    By the above discussion, $\tfrac{1}{\sqrt{s}}A(T_{s,\sigma})\sse A(B_{1,2} \cap \Sigma_\sigma^\mu)$. It is only left to utilize a well-known property of covering numbers. For convenience, let us provide the argument here: By what we just argued, a $\frac{t}{2}$-net $\Theta$ for $A(B_{1,2} \cap \Sigma_\sigma^\mu)$ is also a set with the property that every element in $\tfrac{1}{\sqrt{s}}A(T_{s,\sigma})$ has an element in $\Theta$ at a distance at $\frac{t}{2}$ from it. Hence, it is almost a $\tfrac{t}{2}$-net , but not quite: $\Theta$ must not necessarily consist of elements of $\tfrac{1}{\sqrt{s}}A(T_{s,\sigma})$. But, it is clear that if we define $\Theta'$ as the subset of $\Theta$ that are at most a distance $\sfrac{t}{2}$ from $\tfrac1{\sqrt s} A(T_{s,\sigma})$, $\Theta'$ still covers $\tfrac1{\sqrt s} A(T_{s,\sigma})$. Now we can replace all elements of $\Theta'$ with a $t/2$-close point of $\tfrac1{\sqrt s} A(T_{s,\sigma})$ and, by triangle inequality, arrive at a $t$-net.
    The statement follows.
\end{proof}

Now let us construct the aforementioned random variables $\underline{v}$ which will help us estimate $\calN(A(B_{1,2} \cap \Sigma_\sigma^\mu), \norm{\,\cdot \,}_X, t)$. First, take $\Gamma$ be an $\tfrac{t}{2(1+\delta_\sigma(A))^{\sfrac{1}{2}}}$-net for the set $T_{\sigma} \sse \K^n$ with respect to the $\ell_2$ norm. By the same argument as in the proof of Lemma \ref{lem:smallt}, we know that $\Gamma$ can be chosen so that its cardinality obeys
    \begin{align} \label{eq:cardbound}
        \abs{\Gamma} \leq \binom{n}{\sigma}\left(1+\tfrac{4(1+\delta_\sigma(A))^{\sfrac{1}{2}}}{t}\right)^{2\sigma}.
    \end{align}
Now let $v \in U$ be arbitrary, with $v = A\check{v}$. For each $\check{v}_k$, choose an $\check{u}_k \in \Gamma$ with
    \begin{align*}
        \left\|{\tfrac{\check{v}_k}{\norm{\check{v}_k}}-\check{u}_k}\right\|\leq \tfrac{t}{2(1+\delta_\sigma(A))^{\sfrac{1}{2}}}.
    \end{align*}
    This is possible due to the fact that $\frac{\check{v}_k}{\norm{\check{v}_k}} \in T_\sigma$. Notice that we can make sure that $\supp \check{v}_k = \supp \check{u}_k$, by the construction of the net of $T_\sigma$. Now let us define the random variable $Z$ through 
    \begin{align*}
        \prb{Z = e_k \otimes A\check{u}_k} = \norm{\check{v}_k}, \quad \prb{Z=0} = 1- \norm{v}_{1,2}.
    \end{align*}
    Letting $Z^{(r)}$, $r \in [M]$ (for some $M$ later to be chosen) be independent copies of $Z$, we now finally define
    \begin{align*}
        \underline{v} = \tfrac{1}{M} \sum_{r \in [M]} Z^{(r)}.
    \end{align*}
    Across all possible values of $v$, the variables $\underline{v}$ can only take on a discrete set of values with a certain size. Let us record that for future reference.
    \begin{lem} \label{lem:net}
        There is a set $S$ with
        \begin{align*}
            \abs{S} \leq \left(2\mu \binom{n}{\sigma}\left(1+\tfrac{4(1+\delta_\sigma(A))^{\sfrac{1}{2}}}{t}\right)^{2\sigma}\right)^M
        \end{align*}
        so that 
        \begin{align*}
            \prb{ \forall v \in A(B_{1,2} \cap \Sigma_\sigma^\mu): \, \underline{v} \in S} = 1.
        \end{align*}
    \end{lem}
    \begin{proof}
        For all $v$, the corresponding random variable $Z$ takes values in the set
    \begin{align*}
        \set{0} \cap \set{ e_k \otimes Au \, \vert \, u \in \Gamma, k \in [\mu]},
    \end{align*}
    which has a cardinality $ 1+ \mu \cdot\abs{\Gamma}\leq 2\mu \cdot\abs{\Gamma}$. 
    This holds true for each copy $Z^{(r)}$. Hence, the array $(Z^{(0)}, \dots, Z^{(M-1)})$ can only take on one of $ 2(\mu \cdot \abs{\Gamma})^M$ different values, which extends to $\underline{v}= \tfrac{1}{M} \sum_{r\in [M]} Z^{(r)}$. The bound \eqref{eq:cardbound} on $\abs{\Gamma}$ now yields the claim.
    \end{proof}

It is now left to estimate $\erw{\norm{\underline{v}-v}}_X$. Let us first show that $\erw{\underline{v}}$ is close to $v$.
\begin{lem} \label{lem:erw}
     \begin{align*}
        \norm{\erw{\underline{v}}-v}_X \leq \tfrac{t}{2}
    \end{align*}
\end{lem}
\begin{proof}
    By the definition of $Z$, we have
    \begin{align*}
        \erw{\underline v} = \erw{Z} = \sum_{k \in [\mu]} (e_k \otimes A\check{u}_k) \prb{Z = e_k \otimes A \check{u}_k} = \sum_{k \in [\mu]} e_k \otimes \norm{\check{v}_k}A\check{u}_k.
    \end{align*}
    Let us write $u_k = A\check{u}_k$. The above implies
    \begin{align*}
        \norm{\erw{Z}-v}_X^2 &= \sum_{i \in [m]} \sup_{\theta \in [0,2\pi]} \abs{\sum_{k \in [\mu]}(v_k(i)-\norm{\check{v}_k}u_k(i)) e^{\iota k \theta}}^2 \\
        &\leq \sum_{i \in [m]}  \left( \sum_{k \in [\mu]}\abs{v_k(i)-\norm{\check{v}_k}u_k(i)} \right)^2.
    \end{align*}
    Writing the elements of the inner sums as $\norm{\check{v}_k}^{\sfrac{1}{2}} \cdot \abs{\tfrac{v_k(i)}{\norm{\check{v}_k}^{\sfrac{1}{2}}}-\norm{\check{v}_k}^{\sfrac{1}{2}}u_k(i)}$ and utilizing the Cauchy-Schwarz inequality yields
    \begin{align*}
        \left( \sum_{k \in [\mu]}\abs{v_k(i)-\norm{\check{v}_k}u_k(i)} \right)^2 \leq \left(\sum_{k \in [\mu]} \norm{\check{v}_k}\right) \cdot \left(\sum_{k \in [\mu]}  \abs{\tfrac{v_k(i)}{\norm{\check{v}_k}^{\sfrac{1}{2}}}-\norm{\check{v}_k}^{\sfrac{1}{2}}u_k(i)}^2\right)
    \end{align*}
    for each $i$. Taking the sum over $i$, and utilizing that $\check{v}_k\in B_{1,2}$, we arrive at
    \begin{align*}
        \norm{\erw{Z}-v}_X^2 &\leq  \sum_{i \in [m]}\sum_{k \in [\mu]}\abs{\tfrac{v_k(i)}{\norm{\check{v}_k}^{\sfrac{1}{2}}}-\norm{\check{v}_k}^{\sfrac{1}{2}}u_k(i)}^2 =  \sum_{k \in [\mu]} \norm{\check{v}_k} \cdot \left\|{\tfrac{v_k}{\norm{\check{v}_k}} - u_k }\right\|^2\,.
    \end{align*}
    Now, for each $k$, we have
    \begin{align*}
        \left\|\tfrac{v_k}{\norm{\check{v}_k}} - u_k \right\| = 
        \left\|{A\left(\tfrac{\check{v}_k}{\norm{\check{v}_k}} - \check{u}_k\right) }\right\| \leq (1+\delta_\sigma(A))^{\sfrac{1}{2}}
        \left\|{ \tfrac{\check{v}_k}{\norm{\check{v}_k}} - \check{u}_k}
        \right\|
        \leq \frac{(1+\delta_\sigma(A))^{\sfrac12}\, t}%
        {2(1+\delta_\sigma(A))^{\sfrac{1}{2}}} = \tfrac{t}{2}.
    \end{align*}
    We utilized that $\check{v}_k$ and $\check{u}_k$ are $\sigma$-sparse with the same support, and that $\Gamma$ is a net for $T_\sigma$. Consequently,
    \begin{align*}
        \sum_{k \in [\mu]} \norm{\check{v}_k} \cdot \left\|{\tfrac{v_k}{\norm{\check{v}_k}} - u_k }\right\|^2 \leq \tfrac{t^2}{4}  \sum_{k \in [\mu]} \norm{\check{v}_k}  \leq \tfrac{t^2}{4}\,.
    \end{align*}
\end{proof}

    The above lemma implies that it suffices to study $\erw{\norm{\underline{v}-\erw{\underline{v}}}_X}$ to get the bound we want. This is the purpose of the next lemma.
    \begin{lem} \label{lem:dev}
        Under the assumption that $M \geq 2\log(\mu m)$, we have
        \begin{align*}
            \erw{\left\|{\underline{v}-\erw{\underline{v}}}\right\|_X} \leq 9\pi \sqrt{\frac{\log(8\mu m)(1+\delta_\sigma(A))}{M}}
        \end{align*}
    \end{lem}

\begin{proof}
    First, a symmetrization argument (see e.g \cite[Lemma 8.4]{FouRau2013}) yields that
    \begin{align*}
        \erw{\norm{\underline{v}-\erw{\underline{v}}}_X}  = \erw{\left\|{\tfrac{1}{M} \sum_{r \in [M]} Z^{(r)}-\erw{Z^{(r)}}}\right\|_X} \leq  \mathbb{E}\left(\tfrac{2}{M}\left\|{\sum_{r \in [M]}\varepsilon_r Z^{(r)}}\right\|_X\right),
    \end{align*}
    where $\varepsilon_r$ are independent Rademacher random variables. To estimate this expected value, let us for each fixed $i \in [m]$ and $\theta \in [0,2\pi]$ consider the variable
    \begin{align*}
        Y_i(\theta) = \tfrac{2}{M} \sum_{r\in [M]} \sum_{k \in [\mu]}\varepsilon_r Z^{(r)}_k(i)e^{\iota k \theta}.
    \end{align*}
    $Y_i(\theta)$ is is a random variable of the form $\sum_{r \in [M]} \beta_r \varepsilon_r$, with 
    \begin{align*}
        \beta_r = \tfrac{2}{M}\sum_{k \in [\mu]} Z^{(r)}_k(i)e^{\iota k \theta}.
    \end{align*}
    For each instance of $Z^{(r)}$, $Z^{(r)}_k$ is only non-zero for at most one index $k = k_r$. Consequently,
    \begin{align*}
        \abs{\beta_r} = \tfrac{2}{M}\abs{Z^{(r)}_{k_r}(i)e^{\iota k \theta}} = \tfrac{2}{M}\abs{A\check{u}_{k_r}(i)},
    \end{align*}
    where we used the definition of $Z$ in the final step. (Note that if $Z^{(r)}=0$, we can define $\check{u}_{k_r}=0$). By the Hoeffding inequality for Rademacher sums \cite[Lem. 8.8]{FouRau2013},
    \begin{align*}
        \prb{\abs{ \sum_{r\in [M]}\epsilon_r \beta_r}\geq \alpha \norm{\beta}} \leq 2e^{-\sfrac{\alpha^2}{2}},
    \end{align*}
   were $\beta \in \C^M$ is arbitrary but fixed,
    we obtain
    \begin{align} \label{eq:oneevent}
        \mathbb{P}_\varepsilon \left( \abs{Y_i(\theta)} \geq \alpha \Big(\tfrac{4}{M^2}\!\sum_{r\in [M]}\abs{A\check{u}_{k_r}(i)}^2\Big)^{\sfrac{1}{2} }\right) \leq 2e^{-\sfrac{\alpha^2}{2}},
    \end{align}
    where the $\mathbb{P}_\varepsilon$ means that the probability is with respect to the draw of the $\varepsilon$, i.e. conditioned on the draw of the $Z^{(r)}$. 
    
    Now let us transform this bound into a bound uniform over $\theta \in [0,2\pi]$ and $i\in [m]$. We have
    \begin{align*}
        \abs{Y_i(\theta)-Y_i(\theta')} &\leq \tfrac{2}{M} \sum_{r\in [M]} \sum_{k \in [\mu]}\abs{\varepsilon_r} \abs{Z^{(r)}_k(i)}\vert{e^{\iota k \theta}-e^{\iota k \theta'}\vert} \\
        &\leq \tfrac{2}{M} \sum_{r\in [M]} \sum_{k \in [\mu]} \abs{Z^{(r)}_k(i)}\abs{k(\theta- \theta')}.
    \end{align*}
     Also, by the same argument as before, the sums over $k$ actually only consist of one element, and that element is $\abs{A\check{u}_{k_r}(i)}\abs{k_r(\theta- \theta')}$. Since $\abs{k_r(\theta- \theta')} \leq \mu \abs{\theta-\theta'}$, we conclude that the above is smaller than
    \begin{align*}
       \tfrac{2}{M} \sum_{r\in [M]} \abs{A\check{u}_{k_r}(i)} \mu \abs{\theta-\theta'}\,.
    \end{align*}
    Therefore, if $I_\mu = \set{\tfrac{k}{\mu}\cdot 2\pi \, \vert \, k \in [\mu]}$, we have
    \begin{align} \label{eq:netsup}
        \sup_{\theta \in [0,2\pi]}\abs{Y_i(\theta)} \leq \sup_{\theta \in I_\mu} \abs{Y_i(\theta)} + \tfrac{2\pi}{M} \sum_{r\in [M]} \abs{A\check{u}_{k_r}(i)} \leq \sup_{\theta \in I_\mu} \abs{Y_i(\theta)} + \tfrac{2\pi}{\sqrt{M}} \left(\sum_{r\in [M]} \abs{A\check{u}_{k_r}(i)}^2\right)^{\sfrac{1}{2}}\,.
    \end{align}
    This equation together with a union bound over $i \in [m]$ and $\theta \in I_\mu$ in \eqref{eq:oneevent}  for $M \geq \alpha$ yields
    \begin{align*}
        \mathbb{P}_\varepsilon\left( \sup_{i \in [m],\theta \in [0,2\pi]} \abs{Y_i(\theta)} \geq (\tfrac{2\alpha}{M}+\tfrac{2\pi}{\sqrt{M}} )\Big(\sum_{r\in [M]}\abs{A\check{u}_{k_r}(i)}^2\Big)^{\sfrac{1}{2} }\right) \leq 2\mu m e^{-\sfrac{\alpha^2}{2}}.
    \end{align*}
    Now, utilizing  the inequality $\sfrac{3\pi \alpha}{M} \geq \sfrac{2\alpha}{M} + \frac{2\pi}{\sqrt{M}}$ which for $\alpha \geq \sqrt{M}$, the above yields
    \begin{align*}
          \mathbb{P}_\varepsilon\left( \sup_{i \in [m],\theta \in [0,2\pi]}  \abs{Y_i(\theta)} \geq \tfrac{3\pi\alpha}{M}\Big(\sum_{r\in [M]}\abs{A\check{u}_{k_r}(i)}^2\Big)^{\sfrac{1}{2} }\right) \leq 2\mu m e^{-\sfrac{\alpha^2}{2}}\,.
    \end{align*}
    Assuming the counter event, we have
    \begin{align*}
        \left\|{\frac{1}{M}\sum_{r\in [M]} \varepsilon_r Z^{(r)}}\right\|_X^2 &= \sum_{i \in [m]}\sup_{\theta \in [0,2\pi]} \abs{\sum_{k \in [\mu]}\tfrac{1}{M}\sum_{r\in [M]}\varepsilon_r Z_k^r(i)e^{\iota k \theta}}^2 \\
        &= \sum_{i \in [m]}\sup_{\theta \in [0,2\pi]} \abs{Y_i(\theta)}^2 
        \leq \tfrac{9\pi^2\alpha^2}{M^2} \sum_{i \in [m]}\sum_{r\in [M]}\abs{A\check{u}_{k_r}(i)}^2 
        \\&= \tfrac{9\pi^2\alpha^2}{M^2} \sum_{r\in [M]} \norm{A\check{u}_{k_r}}^2 \leq \tfrac{4\pi^2\alpha^2}{M} (1+\delta_\sigma(A))\,,
    \end{align*}
    where we in the final step utilized that $\check{u}_{k_r}\in T_\sigma$. Hence, for $\alpha \geq \sqrt M$
    \begin{align*}
         \mathbb{P}_\varepsilon\left( \left\|{\frac{1}{M}\sum_{r\in [M]} \varepsilon_r Z^{(r)}}\right\|_X  \geq \tfrac{3\pi\alpha}{\sqrt{M}}(1+\delta_\sigma(A))^{\frac{1}{2}}\right) \leq 2\mu m e^{-\sfrac{\alpha^2}{2}}\,.
    \end{align*}
    This tail-bound can be turned into a bound on the expected value by elementary methods. For convenience, let us just invoke \cite[Prop 7.14]{FouRau2013}. All of its assumptions are met, in particular since we have assumed that  $M\geq 2\log(2\mu m)$. It yields, together with some elementary algebra, the bound
    \begin{align*}
        \mathbb{E}_\varepsilon\left( \left\|{\frac{2}{M}\sum_{r\in [M]} \varepsilon_r Z^{(r)}}\right\|_X\right) &\leq   9\pi \sqrt{\frac{\log(8\mu m)(1+\delta_\sigma(A))}{M}}
    \end{align*}
    where we used that $\check{u}_{k_r} \in T_\sigma $ for all $r$. The theorem of Fubini now yields the claim.
    \end{proof}

Now let us turn the above bounds into a bound of the covering number of $\calN(B_{A(B_{1,2}\cap \Sigma^\mu_\sigma)}, \norm{\, \cdot \,}_{X}, t)$.

\begin{proposition} \label{lem:bigt}
    For $t \leq 2\cdot 81\pi^2$, we have 
    \begin{align*}
        \log(\calN(B_{A(B_{1,2}\cap \Sigma^\mu_{{\sigma}},\sigma)}, \norm{\, \cdot \,}_{X}, t)) \leqsim \tfrac{\log(2\mu m)(1+\delta_\sigma(A))}{t^2}\left( \log(\mu) + \sigma\left( \log(n) +  \log(1 + \tfrac{4(1+\delta_\sigma(A))^{\sfrac{1}{2}}}{t})\right)\right).
    \end{align*}
\end{proposition}
\begin{proof}
    Let us choose
    \begin{align*}
         M =\left\lceil\tfrac{4\cdot 81\pi^2}{t^2}\log(8\mu m)(1+\delta_\sigma(A))\right\rceil.
    \end{align*}
    Then, due to the assumption  $t \leq 2\cdot 81\pi^2$, $M \geq 2\log(2\mu m)$. Hence, Lemma \ref{lem:dev} is applicable. Together with Lemma \ref{lem:erw} and the triangle inequality, we get
    \begin{align*}
        \erw{\norm{v-\underline{v}}} \leq t.
    \end{align*}
   Hence, at least one of the values $\underline{v}$ can take on must lie closer to $v$ than $t$, so that the set $S$ described in Lemma \ref{lem:net} forms a $t$-net for $A(B_{1,2}\cap \Sigma^{{\mu}}_{{\sigma}})$. Consequently
    \begin{align*}
        \log(\calN(U,\norm{\, \cdot \,}_X,t)) &\leq \log(\abs{S}) \leqsim M \left(\log(2\mu) + \sigma \log(n) + \sigma \log\left(1+\tfrac{4(1+\delta_\sigma(A))^{\sfrac{1}{2}}}{t}\right) \right)\\
        &  \leqsim \tfrac{\log(8\mu m)(1+\delta_\sigma(A))}{t^2}\left( \log(\mu) + \sigma \log(n) + \sigma \log\left(1 + \tfrac{4(1+\delta_\sigma(A))^{\sfrac{1}{2}}}{t}\right)\right),
    \end{align*}
    which is the claim.
    \end{proof}
    
\subsection{Estimating the parameters in Theorem \ref{th:krahmerbound}}
Now we have all the tools we need to bound $\gamma_2(\calA(AT_{s,\sigma}), \norm{\, \cdot \,}_{2\to 2})$. For convenience, let us begin by collecting all our previous estimates.
\begin{cor} \label{cor:covering}
    Define 
    \begin{align*}
        \varphi_{0}(t) &= s\log\mu + s\sigma \log(n) + 2s\sigma \log(1+\tfrac{4}{t}) \\
        \varphi_{1}(t) & = \tfrac{\log(8\mu m)}{t^2}\left( \log(\mu) + \sigma \log(n) + \sigma \log(1 + \tfrac{16}{t})\right), \quad
    \end{align*}
    It holds that 
   \begin{align*}
        \log(\calN(\calA(AT_{s,\sigma}), \norm{\, \cdot \,}_{2\to 2}, t)) &\leqsim \varphi_0( \sqrt{\tfrac{\mu}{s}}(1+\delta_\sigma(A))^{-\sfrac{1}{2}} t) \\
        \text{and}\quad
        \log(\calN(\calA(AT_{s,\sigma}), \norm{\, \cdot \,}_{2\to 2}, t)) &\leqsim \varphi_1(\sqrt{\tfrac{\mu}{s}}(1+\delta_\sigma(A))^{-\sfrac{1}{2}} t) \quad\text{{for}}\  t\leq  \tfrac{8 \cdot 81\sqrt{s}}{\sqrt{\mu}}\,.
   \end{align*}
 
\end{cor}
\begin{proof}
The first line is just a paraphrasing of Lemma \ref{lem:smallt}. As for the second, we first use elementary properties of covering numbers (see e.g. \cite[App. C.2]{FouRau2013}) to argue that 
    \begin{align*}
         \calN(\calA(AT_{s,\sigma}), \norm{\, \cdot \,}_{2\to 2}, t)  &=   \calN(\tfrac{1}{\sqrt{s}}\calA(AT_{s,\sigma}), \norm{\, \cdot \,}_{2\to 2}, \tfrac{t}{\sqrt{s}}) \stackrel{\ref{lem:distances}}{\leq} \calN(\tfrac{1}{\sqrt{s}}AT_{s,\sigma}, \tfrac{2}{\sqrt{\mu}}\norm{\, \cdot \,}_{X}, \tfrac{t}{\sqrt{s}})  \\
         &= \calN(\tfrac{1}{\sqrt{s}}AT_{s,\sigma}, \norm{\, \cdot \,}_{X}, \tfrac{t\sqrt{\mu}}{2\sqrt{s}}) \leq  \calN(A(B_{1,2} \cap \Sigma_\sigma^\mu), \norm{\,\cdot \,}_X, \tfrac{t\sqrt{\mu}}{4\sqrt{s}}).  
    \end{align*}
    where we used Lemma \ref{lem:l12ball} in the final step. Now we apply Lemma \ref{lem:bigt}. Note that the bound given there is applicable for $\tfrac{t\sqrt{\mu}}{4\sqrt{s}}\leq 2\cdot 81\pi^2$, as assumed in the lemma.
\end{proof}
Using the above, it is now straight-forward to bound the parameters relevant for Theore \ref{th:krahmerbound}
\begin{cor}  \label{cor:serious_bounds}
Suppose that $\delta_\sigma(A)<1$. Then,
    \begin{align*}
        d_{2\to2}(\calA(AT_{s,\sigma}),\norm{\cdot}_{2\to2}) &\leq \sqrt{\tfrac{s}{\mu}}(1+\delta_\sigma(A))^{\sfrac{1}{2}}, \quad d_{F}(\calA(AT_{s,\sigma}),\norm{\cdot}_{2\to2}) \leq (1+\delta_\sigma(A))^{\sfrac{1}{2}} \\
        \gamma_2(\calA(AT_{s,\sigma}), \norm{\cdot}_{2\to2}) &\leqsim \tfrac1{\sqrt\mu}(1+\delta_\sigma(A))^{\sfrac{1}{2}}(s\log(s)^2\log(\mu m)\log(\mu) + s\sigma\log(n))^{\sfrac{1}{2}}
    \end{align*}
\end{cor}
\begin{proof}
    The statement about $ d_{2\to2}(\calA(AT_{s,\sigma}),\norm{\cdot}_{2\to2})$ follows directly from Lemma \ref{lem:distances}. As for the Frobenius norm statement, we have for $w\in AT_{s,\sigma}$ with representation $\check{w}\in T_{s,\sigma}$, we have
    \begin{align*}
        \norm{\calA(w)}_F^2 = \tfrac{1}{\mu} \sum_{k,j \in [\mu], i \in [n]} \abs{w_{j-k}(i)}^2 = \norm{w}^2 \leq (1+\delta_\sigma(A))\norm{\check{w}}^2,
    \end{align*}
    which implies the bound.

    To bound the $\gamma_2$-functional, let us introduce the notation
    \begin{align*}
        \kappa = \sqrt{\tfrac{s}{\mu}}(1+\delta_\sigma(A))^{\sfrac{1}{2}}.
    \end{align*}
    The aforementioned Dudley bound implies
    \begin{align*}
        \gamma_2(\calA(AT_{s,\sigma}),\norm{\cdot}_{2\to 2} &\leq \int_0^\kappa \log(\calN(\calA(AT_{s,\sigma}), \norm{\cdot}_{2\to 2},t))^{\sfrac{1}{2}} \, \dd t,
    \end{align*}
    where we in particular utilized that $d_{2\to2}(\calA(AT_{s,\sigma}),\norm{\cdot}_{2\to2}) \leq \kappa$. Now, Lemma \ref{cor:covering} implies that 
    \begin{align} \label{eq:covbound}
        \log(\calN(\calA(AT_{s,\sigma}), \norm{\cdot}_{2\to 2},t))^{\sfrac{1}{2}} \leq \min(\varphi_0(\kappa^{-1}t), \varphi_1(\kappa^{-1}t))^{\sfrac{1}{2}}.
    \end{align}
    for all $t \leq  \tfrac{8 \cdot 81\sqrt{s}}{\sqrt{\mu}}$. However, since trivially
    \begin{align*}
        \kappa = \sqrt{\tfrac{s}{\mu}}(1+\delta_\sigma(A))^{\sfrac{1}{2}} \leq \tfrac{\sqrt{2s}}{\sqrt{\mu}} \le \tfrac{8 \cdot 81\sqrt{s}}{\sqrt{\mu}}
    \end{align*}
  the bound is valued for all $t\leq \kappa$. Hence,
    \begin{align*}
        \int_0^\kappa \log(\calN(\calA(AT_{s,\sigma}), \norm{\cdot}_{2\to 2},t))^{\sfrac{1}{2}} \, \dd t &\leq \int_0^\kappa \min(\varphi_0(\kappa^{-1}t), \varphi_1(\kappa^{-1}t))^{\sfrac{1}{2}} \dd t \\
        &= \kappa \int_0^1 \min(\varphi_0(t), \varphi_1(t))^{\sfrac{1}{2}} \dd t.
    \end{align*}
    Now, let's split the integral to one from $0$ to  $\tfrac{1}{\sqrt{s}}$, and one from $\tfrac{1}{\sqrt{s}}$ to $1$. For the first integral, we have
    \begin{align*}
        \kappa \int_0^1 \min(\varphi_0(t), \varphi_1(t))^{\sfrac{1}{2}} \dd t \leqsim \kappa  \int_0^{\sfrac{1}{\sqrt{s}}} \varphi_0(t)^{\sfrac{1}{2}} \dd t &\leqsim  \tfrac{\kappa}{\sqrt{s}}\left(s\sigma \log(n) + s\log(\mu)\right)^{\sfrac{1}{2}} 
    \end{align*}
    simply because the integral of $\log(1+\tfrac{4}{t})^{\sfrac{1}{2}}$ from $0$ to $1$ converges. We may furthermore estimate
    \begin{align*}
        \tfrac{\kappa}{\sqrt{s}}\left(s\sigma \log(n) + s\log(\mu)\right)^{\sfrac{1}{2}} 
        &= \tfrac1{\sqrt{\mu}}(1+\delta_\sigma(A))^{\sfrac{1}{2}}\left(s\sigma \log(n) + s\log(\mu)\right)^{\sfrac{1}{2}} \\  &\leqsim\tfrac1{\sqrt{\mu}}(1+\delta_\sigma(A))^{\sfrac{1}{2}} (s\log(s)^2\log(\mu m)\log(\mu) + s\sigma\log(n))^{\sfrac{1}{2}}
    \end{align*}As for the second one, note that
    \begin{align*}
        \varphi_1(t)^{\sfrac{1}{2}} \leq \tfrac1t \left(2\log(2\mu m)\left(\log(\mu)+\sigma \log(n) + \sigma \log(1+8\sqrt{\sigma})\right)\right)^{\sfrac{1}{2}}
    \end{align*}
    for $t\geq \tfrac{1}{\sqrt{s}}$, and $\int_{\sfrac{1}{\sqrt{s}}}^1 \tfrac{1}{t}\dd t = \log(\sqrt{s})$.  Hence,
    \begin{align*}
        \kappa \int_{\sfrac{1}{\sqrt{s}}}^1 \varphi_1(t)^{\sfrac{1}{2}} \dd s 
        &\leq \tfrac{\sqrt{s}}{\sqrt\mu}(1+\delta_\sigma(A))^{\sfrac{1}{2}}\log(s)\left(2\log(2\mu m)\left(\log(\mu)+\sigma \log(n) + \sigma \log(1+8\sqrt{\sigma})\right)\right)^{\sfrac{1}{2}} \\
        &\leqsim \tfrac1{\sqrt{\mu}}(1+\delta_\sigma(A))^{\sfrac{1}{2}} (s\log(s)^2\log(\mu m)\log(\mu) + s\sigma\log(n))^{\sfrac{1}{2}}\,.
    \end{align*}
    The claim has been proven.
\end{proof}

\subsection{Conclusion}

We now have all the tools we need to prove our main theorem. Let us begin by bounding $\Delta_{s,\sigma}(\mathcal{H})$.

\begin{theorem} \label{th:proto_main_result}
   Under our assumptions,
    \begin{align*}
        \prb{ \Delta_{s,\sigma}(\widehat{H})  > (1+\delta_\sigma(A))\delta_0} \leq  2\epsilon.
    \end{align*}
\end{theorem}
\begin{proof}
    Let us define
    \begin{align*}
        \lambda = \tfrac{1}{\mu}(s\log(s)^2\log(\mu m)\log(\mu) + s\sigma \log(n)).
    \end{align*}
    Our assumption on $\mu$ states that
    \begin{align*}
        \lambda \leqsim \delta_0^2 \min(1,\log( \epsilon^{-1})^{-1}).
    \end{align*}
    This together with 
    Corollary \ref{cor:serious_bounds} then states that, with the notation used in Theorem \ref{th:krahmerbound},
    \begin{align*}
        E &\leqsim  (1+\delta_\sigma(A))( \lambda +\sqrt{\lambda}) \leqsim (1+\delta_\sigma(A))\delta_0 \\
        V &\leqsim (1+\delta_\sigma(A))(\lambda + \sqrt{\lambda}) \leqsim  (1+\delta_\sigma(A))\log( \epsilon^{-1})^{-\sfrac{1}{2}}\delta_0, \\
        U &\leqsim (1+\delta_\sigma(A)) \tfrac{s}{\mu} \leq (1+\delta_\sigma(A))\lambda \leq \delta_0 \log( \epsilon^{-1})^{-1}
    \end{align*}
    Therefore, said theorem implies
    \begin{align*}
        \prb{ \Delta_{s,\sigma}(\widehat{H})  > (1+\delta_\sigma(A))\delta_0} \leq 2\exp(-\log(\epsilon^{-1})) = 2\epsilon,
    \end{align*}
    which was the claim.
\end{proof}

We may now  conclude the proof.
\begin{proof}[Proof of Theorem \ref{th:main_result}]
    First, Lemma \ref{lem:lifting}  and Lemma \ref{lem:factorization}  imply that
    \begin{align*}
      \delta_{s,\sigma}(\calC) \stackrel{\ref{lem:lifting}}{=} \delta_{s,\sigma}(\calH) \stackrel{\ref{lem:factorization}}{\leq} \Delta_{s,\sigma}(\widehat{\calH}) + \delta_\sigma(A) + \Delta_{s,\sigma}(\widehat{\calH}) \delta_\sigma(A).
    \end{align*}
    Now, the lower bound on the number of measurements \ref{eq:mubound} together with Theorem \ref{th:proto_main_result} implies that
    \begin{align*}
        \prb{{\Delta}_{(s,{\sigma})}(\widehat{\calH}) \geq \delta_0 (1+\delta_\sigma(A))}  \leq \epsilon.
    \end{align*}
    Thus, with a probability larger than $1-\epsilon$,
    \begin{align*}
        \delta_{(s,\sigma)}(\calC) \leq \delta_0(1+\delta_\sigma(A)) + \delta_\sigma(A) + \delta_0\delta_\sigma(A)(1+\delta_\sigma(A)= (1+ \delta_\sigma(A))^2 \delta_0 + \delta_\sigma(A),
    \end{align*}
    which was the claim.
    
    \end{proof}

\section{Numerics}\label{sec:numerics}

Let us make a small numerical experiment to test Theorem \ref{th:main_result}. In particular, we want to investigate whether the number of measurements the needed exactly scales as our complexity bound suggests, or if e.g. the $\log(s)$-terms more likely are proof artefacts. 
{Note that such a comparison is inevitably indirect since we numerically observe average case performances while the theoretical results are worst-case bounds over all problem instances.} 

\paragraph{Implementation details} We have implemented our algorithm in the python package \texttt{CuPy} \cite{cupy_learningsys2017}, an open source package for running \texttt{NumPy}-based scripts on  NVIDIA GPUs. We have chosen to do so to utilize the opportunities for parallelization the HiHTP-algorithm allows in this context:
When applying the operator $\calH$ we need to calculate $S_k(Qw_k)$ for all $k \in [\mu]$. Similarly, when applying applying $\calH^*$
\begin{align*}
    \calH^*(y) = \sum_{k \in [\mu]} e_k \otimes (Q^*S_{-k}y),
\end{align*}
$(Q^*S_{-k}y)$ needs to be calculated for all $k \in [\mu]$. Both these sets of calculations can be done in parallel.

Finally, the application of the thresholding operation also benefits from parallelization -- as discussed in \cite{hiHTP}, the application of the thresholding operation consists in first projecting each block onto the  set of $\sigma$-sparse vectors, and then choosing the $s$ projected blocks with the largest norms. The first step here can again be parallellized over the block dimension.

For solving the least-squares problems in each step of the HiHTP-algorithm, we apply a  conjugated gradient algorithm, which we stop once the $\ell_2$-residual is smaller than 1e-4, or $100$ iterations have passed. 
This is justified since in the regime of HiRIP, the  restricted least squares problems are expected to be well-conditioned.

\paragraph{Experimental setup} In all of the experiments, we choose $A=\id$ and in particular $m=n$. $Q=U$ is chosen as a properly renormalized standard Gaussian matrix. We try to solve the blind deconvolution problem for different values of $s, \sigma $ and $\mu$. The values of the sparsity parameters range $\sigma =5,10,15$ and $s = 1, 2, \dots, 7$. We test three values for $n$ 
($n=50$, $n=170$ and $n=350$) and let $\mu$ be $10,20, \dots, 250$. For each quadruple $(n,\mu,s,\sigma)$ we draw $100$ sparse instances of $b$ and $h$ by choosing a sparse support uniformly at random and filling the non-zero entries with independent normally distributed values; we `measure' them with $\calH$, and try to recover it with the HiHTP algorithm. 
We halt the algorithm once the difference in Frobenius norm of consecutive approximations drops below 1e-6 or after $25$ iterations. We then solve the final restricted least squares problem with a lower residual tolerance (1e-6.5). A success is declared when the final relative error between the approximation and the original value for $h \otimes b$ in Frobenius norm is  smaller than 1e-6. The results are depicted in Figure \ref{fig:recprobs}.

\begin{figure}
    \centering
    \includegraphics[width=.7\textwidth]{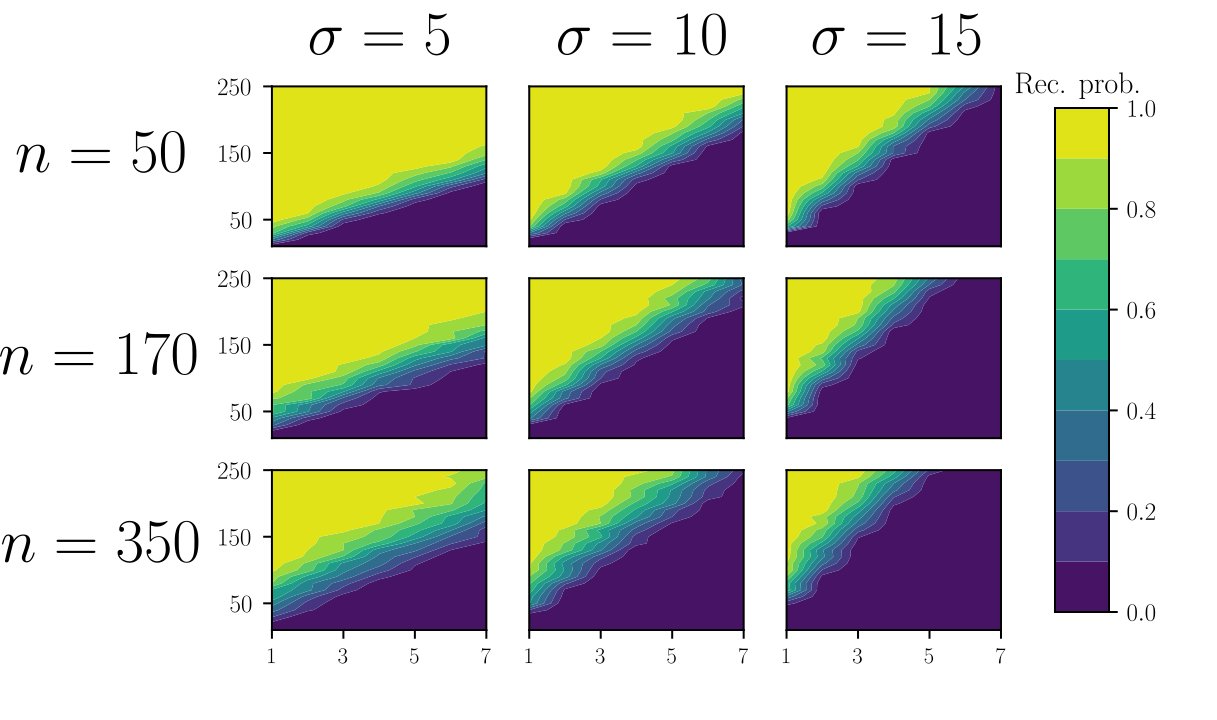}
    \caption{Recovery probability plots over $s=1,2,\dots,7$ and $\mu=10,20,\dots,250$ for different values of $n$ and $\sigma$.
    \label{fig:recprobs}}
\end{figure}

\paragraph{Results} By a simple visual inspection, we see that the number of measurements needed for successful recovery seem to scale linearly with $s$ across all sparsities and values for $n$. We furthermore see that the dependence on $n$ is relatively mild. Hence, the results suggest that the $\log(s)$-terms in our complexity bounds are proof artefacts.

Although not crucial -- $s\sigma$ will for most parameter values be the dominating term in our bound anyhow -- let us also try to test the hypothesis that some terms of our sample complexity bound are spurious a bit more thoroughly. For a tuple of integers $a= (a_0,a_1,a_2)$ and constant $C>0$, let us define the parameter
\begin{align*}
 \lambda_{a,C} = \frac{\mu}{s\log(\mu)^{a_0}\log(\mu)^{a_1}\log(s)^{a_2} + C s\sigma\log(n)}.
\end{align*}
The rationale for defining these parameters is clear: $\lambda_{a,C}$ measures the 'oversampling factor' compared to a sample complexity $\beta_a(s,\sigma,\mu,n) + Cs\sigma \log(n)$, where $$\beta_a(s,\sigma,\mu,n) =s\log(\mu)^{a_0}\log(\mu)^{a_1}\log(s)^{a_2} \sigma \log(n) + C s\sigma\log(n).$$ Using the \texttt{sklearn}-package \cite{scikit-learn}, we now perform logistic regressions of the empirical success probabilites of our data against $\lambda_{a,C}$ for a range of values of $a$ and $C$. The resulting accuracies, when the optimal $C_a$ for $C$ is used for each $a$ are presented in Table \ref{tab:scores}. We see that when using $a=(0,1,0)$, we obtain the highest accuracy. In particular, parameter values with $a_2\neq 0$ lead to worse predictions. Hence, according to our data, the oversampling factor compared to
\begin{align*}
    s\log(\mu n) + C_a s\sigma \log(n)
\end{align*}
is the best predictor of the empirical success of the HiHTP algorithm, suggesting that the bound $\beta_{0,1,0}$ better indicates the regime of success of HiHTP. The differences are small, so one should be cautious read too much into this. However, it is clear that the predictors including additional $\log(s)$-terms in the complexity perform worse. This points to them  being proof artefacts more than anything else in the bound of Theorem \ref{th:main_result}. In Figure \ref{fig:logfits}, the optimal predictor $\lambda_{(0,1,0)}$ is compared with $\lambda_{(1,1,2)}$, which is the one suggested by Theorem \ref{th:main_result}. The figures clearly suggest that the former fits much better to the data.

\begin{table}
    
\end{table}

\begin{figure}
    \centering
    \includegraphics[width=.4\textwidth]{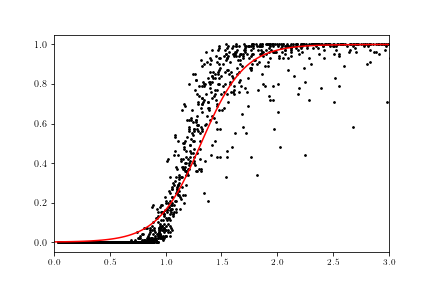}\quad\includegraphics[width=.4\textwidth]{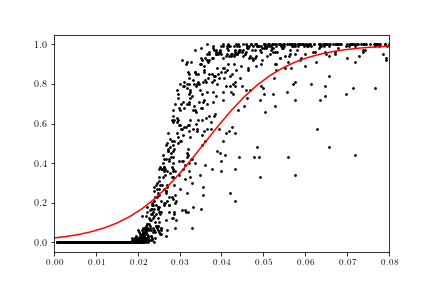}
    \caption{The empirical recovery probabilities and fitted logistic curves plotted against the optimal predictor $\lambda_{(0,1,0)}$(left) and the one suggested by Theorem \ref{th:main_result} $\lambda_{(1,1,2)}$ (right). The vertical axes have been cropped to highlight the phase transition region. }
    \label{fig:logfits}
\end{figure}

\begin{table} \centering
    \begin{tabular}{l|l || l | l ||l | l}
        $\beta_a$ & Accuracy &   $\beta_a$ & Accuracy &   $\beta_a$ & Accuracy  \\
        \hline  $s$ & 91.8\% &  $s\log(s)$ & 91.2\% & 
        $s\log(s)^2$ & 91.1\% \\
        $s\log(\mu)$ & 91.7\% & 
        $s\log(\mu)\log(s)$ & 91.2\% &
        $s\log(\mu)\log(s)^2$ & 91.1\%\\
        $s\log(\mu n)$ & \textbf{92.1\%} & 
        $s\log(\mu n)\log(s)$ & 91.2\%&
        $s\log(\mu n)\log(s)^2$ & 91.1\%\\
        $s\log(\mu n)\log(\mu)$ & 91.7\% & 
        $s\log(\mu n)\log(\mu)\log(s)$ & 90.9\%&
        $s\log(\mu n)\log(\mu)\log(s)^2$ & 90.3\%
    \end{tabular}
    \caption{Accuracy of the predictions resulting from the logistic regressions for different values of $a$. \label{tab:scores}}
\end{table}

\subsection*{Acknowledgement} The authors wish to thank the anonymous reviewers of previous versions of the manuscript for helpful comments and suggestions, leading to a significant improvement of the results. AF acknowledges support from the Wallenberg AI, Autonomous Systems and Software Program (WASP) funded by the Knut and Alice Wallenberg Foundation, and CHAIR. GW is supported by the German Science Foundation (DFG) under grants 598/7-1, 598/7-2, 598/8-1, 598/8-2 and the 6G research cluster (6g-ric.de) supported by the German Ministry of Education and Research (BMBF).

\bibliographystyle{abbrv}
\bibliography{block_model}

\end{document}